\documentclass[11pt]{article}
\usepackage[margin=1in]{geometry}
\usepackage{epsfig}
\usepackage{amssymb,amsmath,amsthm}
\usepackage{mathrsfs}
\usepackage{graphicx}
\usepackage{multirow}
\usepackage{color}
\usepackage{url}

\newcommand{\remove}[1]{}

\newcommand{\ext}{\mathsf{Ext}}
\newcommand{\zo}{\{0,1\}}

\newtheorem{thm}{Theorem}

\theoremstyle{definition}

\newtheorem{theorem}[thm]{Theorem}
\newtheorem{lemma}[thm]{Lemma}
\newtheorem{definition}[thm]{Definition}

\newtheorem{corollary}[thm]{Corollary}

\newtheorem{remark}[thm]{Remark}

\newcommand*\samethanks[1][\value{footnote}]{\footnotemark[#1]}

\title{Leakage-Resilient Non-Malleable Secret Sharing in Non-compartmentalized Models} 

\author{Fuchun Lin\thanks{Division of Mathematical Sciences, School of Physical and Mathematical Sciences, Nanyang Technological University, SG}
\and Mahdi Cheraghchi\thanks{Department of Computing, Imperial College London, UK} 
\and Venkatesan Guruswami\thanks{Computer Science Department, Carnegie Mellon University, USA}  
\and Reihaneh Safavi-Naini\thanks{Department of Computer Science, University of Calgary, CA}
\and Huaxiong Wang\samethanks[1]}

\date{}

\usepackage[hidelinks]{hyperref}

\begin{document}

\maketitle
\begin{abstract} 

Non-malleable secret sharing was recently proposed by Goyal and Kumar in independent tampering and joint tampering models for threshold secret sharing (STOC18) and secret sharing with general access structure (CRYPTO18). The idea of making secret sharing non-malleable received great attention and by now has generated many papers exploring new frontiers in this topic, such as multiple-time tampering and adding leakage resiliency to the one-shot tampering model. Non-compartmentalized tampering model was first studied by Agrawal et.al (CRYPTO15) for non-malleability against permutation composed with bit-wise independent tampering, and shown useful in constructing non-malleable string commitments. In spite of strong demands in application, there are only a few tampering families studied in non-compartmentalized model, due to the fact that compartmentalization (assuming that the adversary can not access all pieces of sensitive data at the same time) is crucial for most of the known techniques. 

We initiate the study of leakage-resilient secret sharing in the non-compartmentalized model.
Leakage in leakage-resilient secret sharing is usually modelled as arbitrary functions with bounded total output length applied to each share or up to a certain number of shares (but never the full share vector) at one time. Arbitrary leakage functions, even with one bit output, applied to the full share vector is impossible to resist since the reconstruction algorithm itself can be used to construct a contradiction. We allow the leakage functions to be applied to the full share vector (non-compartmentalized) but restrict to the class of affine leakage functions. The leakage adversary can corrupt several players and obtain their shares, as in normal secret sharing. The leakage adversary can apply arbitrary affine functions with bounded total output length to the full share vector and obtain the outputs as leakage. These two processes can be both non-adaptive and do not depend on each other, or both adaptive and depend on each other with arbitrary ordering. We use a generic approach that combines randomness extractors with error correcting codes to construct such leakage-resilient secret sharing schemes, and achieve constant information ratio (the scheme for non-adaptive adversary is near optimal).

We then explore making the non-compartmentalized leakage-resilient secret sharing also non-malleable against tampering. We consider a tampering model, where the adversary can use the shares obtained from the corrupted players and the outputs of the global leakage functions to choose a tampering function from a tampering family $\mathcal{F}$. We give two constructions of such leakage-resilient non-malleable secret sharing for the case $\mathcal{F}$ is the bit-wise independent tampering and, respectively, for the case $\mathcal{F}$ is the affine tampering functions, the latter is non-compartmentalized tampering that subsumes the permutation composed with bit-wise independent tampering mentioned above.


\end{abstract}


\section{Introduction}
Secret sharing, introduced independently by Blakley \cite{Blakley} and Shamir \cite{Shamir},
is a fundamental cryptographic primitive with far-reaching
applications; e.g., a major tool in secure multiparty computation (cf.\ \cite{ref:CDB15}). 
The goal in secret sharing
is to encode a secret $\mathsf{s}$ into a number of \emph{shares} 
$\mathsf{c}_1, \ldots, \mathsf{c}_P$ 
that are distributed
among a set $\mathcal{P}=\{1,\ldots,P\}$ of players such that the access to the secret through collaboration of players can be 
accurately controlled. 
An \emph{authorized} subset of players is a set $A \subseteq \mathcal{P}$ such that
the shares with indices in $A$ can be pooled together to reconstruct the secret $\mathsf{s}$.
On the other hand, $A$ is an \emph{unauthorized} subset if the knowledge of
the shares with indices in $A$ reveals no information about the secret. 
The set  of 
 authorized and unauthorized sets define an  access structure,
where the most widely used is the so-called
\emph{threshold} structure. A  threshold 
secret sharing scheme is defined
with respect to a reconstruction threshold $r$ and satisfies the following property:
Any set $A \subseteq \mathcal{P}$ with $|A| < r$ is an unauthorized set and any set $A \subseteq \mathcal{P}$ with $|A| \geq r$ is an authorized set.
Any threshold secret sharing scheme  sharing $\ell$-bit secrets necessarily requires shares of length at least $\ell$, and Shamir's scheme attains this lower bound \cite{ref:Sti92}. The {\em information ratio} defined as the ratio of the maximum share length to the secret length measures the storage efficiency of a secret sharing scheme. 

Non-malleable codes \cite{DzPiWi} proposed 
with applications in tamper-resilient cryptography in mind are codes with a randomized encoder and a deterministic decoder that provide non-malleability guarantee with respect to a  family $\mathcal{F}$ of tampering functions: Decoding the tampered codeword 
yields the original message 
or a value that follows a fixed distribution, where the probability of the first case and the probability distribution in the second case are dictated by the particular tampering function $f\in\mathcal{F}$ alone 
(all probabilities are taken over the randomness of the encoder). 
Intuitively, non-malleable coding prevents the adversary from tampering with the protected message in a message-specific way, which is the essence of {\em non-malleable cryptology} \cite{DDN00}. 
Perhaps the most widely studied tampering model for non-malleability is the {\em compartmentalized} model called the {\em $P$-split state} model, where for a constant integer $P$,  
a tampering function is described by $f=(f_1,\ldots,f_P)$, for arbitrary functions $f_i\colon\{0,1\}^{N/P}\rightarrow\{0,1\}^{N/P}$.
Goyal and Kumar initiated a systematic study of non-malleable secret sharing \cite{STOC18,CRYPTO18} with inspirations from the non-malleable codes. Their study started with the observation that a $2$-split state non-malleable code is a non-malleable $2$-out-of-$2$ (statistical) secret sharing 
(the privacy follows directly from non-malleability in the 2-split state model, see \cite{AggarwalTCC15} for a proof). 
So the $2$-out-of-$2$ case has many constructions
(just to name a few and restrict to information-theoretic security)\cite{onebit,additivecombinatorics,AggarwalTCC15,NMreduction,CGL16,Li17,XinLi18,CL19}. 
Goyal and Kumar~\cite{STOC18} proposed two tampering models
for $r$-out-of-$P$ secret sharing for $P>2$ and any $r\leq P$. 
The {\em independent tampering} model of non-malleable secret sharing is essentially a secret sharing with $P$ players, which is non-malleable with respect to the $P$-split state tampering family. 
The {\em joint tampering} model allows the adversary to group any $r$ shares into two subsets of different size and tamper jointly with the shares within each group but independently across the two groups. 
In the follow up work \cite{CRYPTO18}, non-malleability was generalized to 
secret sharing with general access structures. In the independent tampering model, they constructed a compiler that transforms any plain secret sharing into a non-malleable secret sharing with the same access structure. In the joint tampering model, explicit $P$-out-of-$P$ threshold secret sharing against more powerful adversaries that can group shares into two overlapping subsets, as long as no authorized set is jointly tampered, are constructed.
The idea of making secret sharing non-malleable against tampering has attracted a lot of attention and generated many papers exploring new frontiers in this topic. Srinivasan and Vasudevan \cite{SV18} constructed the first non-malleable secret sharing for $4$-monotone access structures with constant information ratio. 
Badrinarayanan and Srinivasan \cite{BS18} considered a multiple-time tampering model (corresponding to {\em continuous non-malleable codes}) for secret sharing where the tampering adversary can non-adaptively specify a sequence of tampering functions in the independent tampering model and non-malleability guarantee should hold for the whole sequence of tampering (assuming the same reconstruction set). Aggarwal et al. \cite{ADNOPRS18} considered a strengthening of the above multiple-time tampering model that takes into account the subtlety of secret reconstruction in secret sharing. In particular, they allow the tampering adversary to control the secret reconstruction from tampered shares by specifying the reconstruction set in each time (they dub this {\em non-adaptive concurrent reconstruction}). Kumar, Meka, and Sahai \cite{KMS18} initiated the study of leakage-resilient non-malleable secret sharing, where the tampering adversary is allowed to base the choice of tampering on the information about the encoding obtained from leaking every share independently. This defines a stronger type of tampering (than without leakage) because the randomness of the encoder decreases conditioned on the leaked value, which has an effect on non-malleability guarantee (relying on the randomness of the encoder by definition). Faonio and Venturi \cite{FV19} considered a strengthen model that has multiple-time tampering with {\em adaptive} concurrent reconstruction and leakage-resilience, but had to switch to computational security.
See Table \ref{tab: comparison} for a summary of different models.

\begin{table}[h]
    \caption{Comparison of models for the existing LR-SS and NM-SS with $P>2$ players}\label{tab: comparison}
      \begin{center}
       \begin{tabular}{|l|l|l|l|l|l|p{4cm}|}\hline
         Reference      &   Access Structure  &Design Goal&Leakage/Tampering Model\\
        \hline\hline

            \cite{DP07}   &Round complexity based& LR-SS&Independent Leakage (Ind. L.)\\
            \hline
            
          \cite{BDIR18} &$r$-out-of-$P$  &LR-SS &Ind. L. \\ 
           \hline

                           &$2$-out-of-$P$ & LR-SS&Ind. L.\\
          \cite{STOC18}   &$r$-out-of-$P$ & NM-SS&Independent Tampering (Ind. T.)\\
                           &$r$-out-of-$P$ & NM-SS &Joint Tampering (Joint T.)\\ 
           \hline

           \cite{CRYPTO18} & Arbitrary  & NM-SS &Ind. T.\\
                         &$P$-out-of-$P$ & NM-SS &Joint T. \\
           \hline

            \cite{BS18} &Arbitrary (4-monotone) &CNM-SS&Continuous Ind. T. (CNM-SS)\\ 
           \hline

           \cite{ADNOPRS18} & Arbitrary& LR-SS & Ind. L. \\ 
                           & Arbitrary (3-monotone)& CNM-SS &Non-adap. concurrent reconstruct\\ 
          \hline

          \cite{SV18}& $r$-out-of-$P$ &LR-SS&Ind. L. $\leftarrow$ $r-2$ shares \\
                     &  Arbitrary (4-monotone)& NM-SS &Ind. T. \\ 
           \hline

          \cite{KMS18}& Arbitrary &CLR-SS&Continuous adap. Joint Leakage\\
                        & Arbitrary & LR-NM-SS  &Ind. T. $\leftarrow$Ind. L.\\ 
           \hline

           \cite{FV19}* &Arbitrary & LR-CNM-SS&Ind. noisy L.\\
                                        &&&Adap. concurrent reconstruct \\ 
           \hline\hline

                           &$r$-out-of-$P$ & LR-SS&Affine L. \ ------\underline{first NComp. L.}\\
      This work      &$r$-out-of-$P$ & LR-NM-SS&Bit-wise Ind. T. $\leftarrow$ Affine L.\\
                           &$r$-out-of-$P$ & LR-NM-SS&NComp.T.$\leftarrow$Affine L. ---\underline{first NComp.T.}\\
                           \hline
     \end{tabular}
     \end{center}
         \footnotesize{Only the features concerning modelling are captured in this table. Shorthands are defined where they first appear in the table. 
         The symbol ``$\leftarrow$'' denotes ``based on''. 
         \cite{FV19}* has a * because it is the only one using computational assumptions. ``NComp.'' is short for ``Non-Compartmentalized''. The view of the leakage adversary in ``Affine L.'' model contains a choice of $r-1$ shares and a bounded length output of a choice of affine functions applied to the full share vector. The choice of the $r-1$ shares and the affine functions can be non-adaptive or adaptive. }
     \end{table}

A leakage-resilient secret sharing scheme hides the secret from an adversary, who in addition to having access to an unqualified set of shares, also obtains some bounded length leakage from all other shares. Leakage-resiliency for secret sharing was in fact studied much earlier than non-malleable secret sharing. Dziembowski and Pietrzak \cite{DP07} developed an intrusion-resilient secret sharing scheme using alternating extractors. Dav{\`{\i}}, Dziembowski and Venturi \cite{DDV10} constructed the first 2-out-of-2 secret sharing scheme that statistically hides the secret even after an adaptive adversary executes a bounded communication leakage protocol on the two shares. The leakage-resilient non-malleable codes in $2$-split state model of Liu and Lysyanskaya \cite{LL12} (computational security) and \cite{AggarwalTCC15} 
are also $2$-out-of-$2$ leakage-resilient secret sharing which also feature non-malleability. 
Recently, as the dual result of \cite{GW17}, which shows that by leaking one bit from each share, the secret of the a Shamir scheme over finite field with characteristic $2$ can be completely reconstructed, Benhamouda, Degwekar, Ishai and Rabin \cite{BDIR18} showed that the Shamir $r$-out-of-$P$ secret sharing
scheme, when the underlying field is of a large prime order and for large values of $r=P-o(\log P)$ is leakage-resilient against a non-adaptive adversary who independently leaks bounded amount of information from each share.
Goyal and Kumar \cite{STOC18,CRYPTO18} constructed a $2$-out-of-$P$ leakage-resilient secret sharing scheme as a building block for their constructions of non-malleable secret sharing. Aggarwal et al. \cite{ADNOPRS18} proposed a construction for general access structure and a new application to leakage-resilient threshold signatures.
Several strengthened leakage-resilient secret sharing models have been proposed. Srinivasan and Vasudevan \cite{SV18} proposed a leakage model for $r$-out-of-$P$ threshold schemes, where the choice of each local leakage function can be based on a choice of $r-2$ shares.  Kumar, Meka, and Sahai \cite{KMS18} proposed a bounded length multiple-round adaptive joint leakage model. The adversary can choose different unauthorized sets of shares to jointly leak from them and output messages multiple times. Adaptive here means that each time the choice of the unauthorized set and the leakage function are based on all previous outputted messages. 
\cite{FV19} (computational assumption) considered a noisy leakage model that, instead of bounding the output length of the leakage functions, bounds the min-entropy of the share conditioned on the output. 
See Table \ref{tab: comparison} for a summary of different models.

In the context of non-malleable codes, Agrawal et.al \cite{AGMPP15} initiated the study of non-compartmentalized tampering models. They considered non-malleability against permutation composed with bit-wise independent tampering, and showed that non-malleable codes in such a tampering model transform non-malleable bit-commitments into a non-malleable string-commitment. They also gave a rate $1$ construction for such non-malleable codes \cite{MajiTCC}. There are a few other non-compartmentalized tampering families studied for non-malleable codes: local functions \cite{localNMC}, affine functions \cite{CL17}, small-depth circuits \cite{smalldepthcircuit} and decision tree \cite{decisiontree}. In particular, the affine tampering model not only includes the permutation composed with bit-wise independent tampering of \cite{AGMPP15,MajiTCC} as a special case, but also captures a much stronger adversary than \cite{AGMPP15,MajiTCC} and local tampering of \cite{localNMC} in that each output bit of a $\mathbb{F}_2$-affine function can depend on all input bits. In this sense, affine functions are arguably the best example of the non-compartmentalized model.

There has not been non-compartmentalized tampering model studied in non-malleable secret sharing. This is partly because currently known constructions of non-malleable secret sharing crucially rely on the tools that only work for compartmentalised models (e.g. independent source extractors and secret sharing schemes). 
Almost all constructions of non-malleable secret sharing take the approach of building a compiler that transforms several plain secret sharing schemes with various extra properties into a non-malleable secret sharing. 
It is not clear how resiliency against a global tampering can be realized using this approach.

Dav{\`{\i}}, Dziembowski and Venturi \cite{DDV10}, apart from constructing the first $2$-out-of-$2$ leakage-resilient secret sharing, proposed a general leakage model called {\em Leakage-Resilient Storage (LRS)}, where there is an upper bound on the total output length and the leakage functions can be chosen from a set $\mathcal{L}$ of functions that is only restricted by its cardinality $|\mathcal{L}|$. The cardinality $|\mathcal{L}|$ can still be exponential in the length of the encoding and functions computable by Boolean circuits of a fixed size was given as an example for this model. 

Again, there has not been non-compartmentalized leakage model studied for leakage-resiliency for secret sharing. 
The leakage for secret sharing is usually modelled as an arbitrary function with bounded output length applied to each share or up to a certain number of shares (but never the full share vector) at one time. Note that arbitrary leakage functions, even with one bit output, applied to the full share vector is impossible to resist since the reconstruction algorithm itself can be used to construct a contradiction. Indeed, a counter example could be the reconstruction algorithm outputting the first bit of the secret. It is not clear how the LRS with $\mathcal{L}$ only restricted by its cardinality $|\mathcal{L}|$ can be realized for secret sharing. 

\smallskip
\noindent
{\bf Our contributions.} 
We take inspiration from the definition of non-malleable (codes) secret sharing and propose a general notion of leakage-resilient secret sharing with respect to a structured family $\mathcal{L}$ of leakage functions and a total output size bound $\beta$, which is a non-negative integer. We call a leakage adversary in this model a {\em $\beta$-bounded $\mathcal{L}$-leakage adversary}. 
We fill the gap left open in current state of leakage-resilient secret sharing by considering a structured non-compartmentalized leakage family $\mathcal{L}$.
In particular, we focus on the family $\mathcal{L}_\mathsf{affine}$ of $\mathbb{F}_2$-affine leakage functions and design leakage-resilient secret sharing schemes against a $\beta$-bounded $\mathcal{L}_\mathsf{affine}$-leakage adversary. We emphasize that each output bit of the leakage function can depend on all input bits, namely, the full share vector.
When the context is clear, we simply call it {\em affine leakage-resilient secret sharing}. 

\smallskip
\noindent
\textbf{Definition} (Informal). \textit{An $r$-out-of-$P$ statistical $\beta$-bounded affine leakage-resilient secret sharing is a $r$-out-of-$P$ statistical secret sharing scheme that is also statistically 
leakage-resilient against a $\beta$-bounded $\mathcal{L}_\mathsf{affine}$-leakage adversary. More concretely,
\begin{enumerate}
\item Correctness: given any $r$ shares, the secret is correctly reconstructed with overwhelming probability, over the randomness of the sharing algorithm.
\item Privacy and Leakage-Resiliency:       
       \begin{itemize}
       \item Non-adaptive adversary: any non-adaptive choice of $r-1$ shares and the output of any non-adaptive choice of affine leakage functions of the full share vector with total output length bounded by $\beta$ are statistically indistinguishable for any pair of distinct secrets.
       \item Adaptive adversary: any adaptive choice of $r-1$ shares and the output of any adaptive choice of affine leakage functions of the full share vector with total output length bounded by $\beta$ are statistically indistinguishable for any pair of distinct secrets. (The choice of the $r-1$ shares and the choice of the affine leakage functions can adaptively depend on each other.)
       \end{itemize}
\end{enumerate}
}

Using the construction of optimal non-adaptive binary secret sharing in \cite{preprint}, we immediately have a non-adaptive $r$-out-of-$P$ statistical secret sharing with asymptotic information ratio $1$. We are able to prove that by shortening the secret by $\beta$ bits, the $r$-out-of-$P$ statistical secret sharing can be made leakage-resilient against a non-adaptive $\beta$-bounded $\mathcal{L}_\mathsf{affine}$-leakage adversary. We then have the following.

\smallskip
\noindent
\textbf{Theorem} (Informal). \textit{There is a non-adaptive $r$-out-of-$P$ statistical $\beta$-bounded affine leakage-resilient secret sharing for any constant $r$ and $P$ with secret length $\ell$ and information ratio $\frac{\ell+\beta+o(\ell)}{\ell}$.}

We note that this information ratio is almost the best one can hope for. Intuitively, any $r$ shares contain the full information about the $\ell$ bits secret, while $r-1$ shares among them do not contain any information. This means that the amount of secret information possible is the upper-bounded by the length of one share. Now there are $\beta$ bits information about these $r$ shares leaked to an unconditional adversary. The upper bound on the amount of secret information must reduce by $\beta$ bits. In other words, an information ratio of $\frac{\ell+\beta}{\ell}$ would be the optimal.

\bigskip

One could use the construction of adaptive binary secret sharing in \cite{preprint} to construct affine leakage-resilient secret sharing. We propose a new construction that have a better information ratio. As a result of independent interest, our construction of adaptive leakage-resilient secret sharing here also gives an adaptive binary secret sharing with improved {\em coding rate} (see {\bf Related works} for more details).

\smallskip
\noindent
\textbf{Theorem} (Informal). \textit{There is an adaptive $r$-out-of-$P$ statistical $\beta$-bounded affine leakage-resilient secret sharing for any constant $r$ and $P$ with secret length $\ell$ and constant information ratio.} 

\bigskip

We extend our affine leakage-resilient secret sharing model to a leakage-resilient non-malleable secret sharing model. We again consider a general tampering family $\mathcal{F}$ that can possibly be non-compartmentalized. We allow the tampering adversary to base the choice of the tampering function $f\in\mathcal{F}$ on any unauthorised set of shares and 
the output of the $\mathcal{L}$-leakage from the full share vector. We call it {\em affine leakage-resilient non-malleable secret sharing}, when the tampering family $\mathcal{F}$ needs not be specified.

\smallskip
\noindent
\textbf{Definition} (Informal). \textit{An adaptive $r$-out-of-$P$ statistical $\beta$-bounded affine leakage-resilient secret sharing is said to be non-malleable with respect to a tampering family $\mathcal{F}$ if the following non-malleability property is satisfied.
}

\textit{Non-malleability: for any up to $r-1$ shares, any $\beta$-bounded $\mathcal{L}_\mathsf{affine}$-leakage adversary, any $\mathcal{F}$-tampering strategy $\sigma$ and any reconstruction set $R$ of size $r$, reconstructing from the set $R$ of the tampered shares yields the original secret or a value that follows a background distribution, where the probability of the first case and the probability distribution in the second case are dictated by the particular leakage adversary, the particular tampering strategy $\sigma$ and the particular reconstruction set $R$ (all probabilities are taken over the randomness of the sharing algorithm).  
}

The first family $\mathcal{F}$ of tampering functions we consider is the family $\mathcal{F}_\mathsf{affine}$ of $\mathbb{F}_2$-affine tampering functions. 
By strengthening one of the building blocks of the adaptive binary secret sharing construction in \cite{preprint} to its ``non-malleable counterpart'' (from an {\em affine extractor} to an {\em affine non-malleable extractor}, see {\bf Overview of constructions} below for more information), we are able to prove that the non-malleability property, in addition to correctness, privacy and leakage-resiliency of affine leakage-resilient secret sharing, is satisfied. This gives us a leakage-resilient non-malleable secret sharing fully in non-compartmentalized model. That is the leakage model is $\mathcal{L}_\mathsf{affine}$ and the tampering model is $\mathcal{F}_\mathsf{affine}$, both are non-compartmentalized. 

\smallskip
\noindent
\textbf{Theorem} (Informal). \textit{There is an adaptive $r$-out-of-$P$ statistical $\beta$-bounded affine leakage-resilient secret sharing for any constant $r$ and big enough $P$ that is non-malleable with respect to $\mathcal{F}_\mathsf{affine}$.}

\smallskip

The above construction in fact proves a reduction from an affine leakage-resilient non-malleable secret sharing with respect to $\mathcal{F}_\mathsf{affine}$ to a special type of randomness extractor (affine non-malleable extractor), an object in pseudo-randomness. We would be able to get explicit secret sharing schemes for any constant $r$ and $P$, once affine non-malleable extractors with better parameters (one that can extract from any constant fraction of entropy) are explicitly constructed.

\bigskip

The second family $\mathcal{F}$ of tampering functions we consider is the family $\mathcal{F}_\mathsf{BIT}$ of {\em Bit-wise Independent Tampering (BIT)} functions. Let $q$ be the size of each share. A function $f\in\mathcal{F}_\mathsf{BIT}$ for a secret sharing with $P$ players is described by $f=(f_1,\ldots,f_{P\log q})$, where $f_i$ is a binary  tampering function belonging to $ \{ \mathsf{Set0}, \mathsf{Set1}, \mathsf{Keep}, \mathsf{Flip}\}$, where $\mathsf{Set0}$ and $\mathsf{Set1}$  set the value of the bit to $0$ and $1$, respectively, and  $ \mathsf{Keep}$ and $ \mathsf{Flip}$ will keep and flip the bit, respectively. For this tampering family, we are able to modify our construction of adaptive affine leakage-resilient secret sharing to also satisfy non-malleability, for any constants $r$ and $P$.

\smallskip
\noindent
\textbf{Theorem} (Informal). \textit{There is an adaptive $r$-out-of-$P$ statistical $\beta$-bounded affine leakage-resilient secret sharing for any constant $r$ and $P$ that is non-malleable with respect to $\mathcal{F}_\mathsf{affine}$.}

\bigskip

Note that since the tampering function $f\in\mathcal{F}$ is chosen based on any up to $r-1$ shares and the output of the leakage function, the non-malleable secret sharing model is in fact not weak even when $\mathcal{F}=\mathcal{F}_\mathsf{BIT}$. In particular, the tampering at the $r-1$ shares chosen by the privacy adversary is similar to joint tampering model, though the tampering at other shares is weaker than the leakage-resilient version of independent tampering model. The modification we have here is similar to the one in the construction for $\mathcal{F}_\mathsf{affine}$. We replace the {\em linear seeded extractor} with a {\em linear seeded non-malleable extractor} (see {\bf Overview of constructions} below for more information). Seeded non-malleable extractors are under scrutiny in the past few years and many good constructions are known (just to name a few \cite{DW07,XinLiIP,DLTWZ14,overcome,Li17,XinLi18}). But as far as we know, only the inner product construction of \cite{XinLiIP} gives a linear seeded non-malleable extractor. We then prove an existence result concerning the linear seeded non-malleable extractors with our required properties (one that can extract from less than one half of the entropy) and leave its explicit construction as an interesting open problem.

 \smallskip
\noindent
{\bf Overview of constructions.} 
A $(t,r,P)$-ramp scheme is defined with respect to two thresholds, $t$ and $r$. The knowledge of any $t$ shares or fewer does not reveal any information about the secret.
On the other hand, any $r$ shares can be used to reconstruct the secret. The subsets of size $\geq t+1$ or $\leq r-1$ shares, may reveal some information about the secret. 
Note that the guarantee  that requires all subsets of participants be either authorized, or unauthorized is no longer be attained when $r-t>1$. 
We state our results in the language of $(t,r,P)$-ramp schemes and the results specialised to threshold secret sharing mentioned above can be recovered by letting $r-t=1$. 

 \bigskip
 \noindent
{\em Affine leakage-resilient secret sharing.} 
An extractor is a function that turns non-uniform distributions (called {\em source}) over the domain  into an almost uniform distribution over the range (smaller in size than the domain). An affine source is a flat distribution on an affine subspace and an extractor for affine sources is called an affine extractor. An extractor is invertible if there is an efficient algorithm that, given an extractor output, samples a pre-image for that output uniformly at random. Very recently, Lin et.al. \cite{preprint} proposed a construction of secret sharing through combining an invertible affine extractor and a linear erasure correcting code.  
In their construction, the secret is the output of the affine extractor. The sharing algorithm first uses the inverter of the extractor to sample a random pre-image for the secret, then encodes the pre-image using the erasure correcting code. 
The key observation is if we start with a uniformly distributed secret, the inverter will output a distribution that is uniform over the domain of the extractor. The privacy analysis is focused on this uniform pre-image. Now this pre-image is further encoded using the erasure correcting code to yield the share vector. But since the erasure correcting code is linear, knowing several components of its codeword and knowing several bits output of an affine function of its codeword amount to putting several linear equations on the uniform pre-image, which is now
flatly distributed on an affine sub-space of the domain of the extractor, hence an affine source. 
If this affine source has enough entropy, then the distribution of the uniform secret conditioned on the adversary's view remains uniform. This means that the adversary's view and the secret are independent and hence privacy is provided. Using this construction, ramp secret sharing families with {\em statistical privacy} and {\em probabilistic reconstruction} over {\em binary} shares can be constructed, given any 
relative privacy threshold $\tau$ and relative reconstruction threshold $\rho$, for arbitrary constants $0\leq\tau<\rho\leq1$. 
Now given a privacy threshold $t$ and a reconstruction threshold $r$ for a ramp scheme with $P$ players, we set $\tau=t/P$ and $\rho=r/P$, and obtain a family of binary ramp schemes with $N$-bit share vector, where $N$ is a multiple of $P$. 
We then divide the $N$-bit share vector into $P$ blocks and call each block a share of a $(t,r,P)$-ramp scheme. 

Our construction of non-adaptive leakage-resilient secret sharing uses the same high level ideas as described above but with a linear {\em seeded} extractor instead of a seedless one. A seeded extractor is a function that takes a second input (called the seed) which is uniform and independent of the source input. The interest in the use of seeded, as opposed to seedless affine, extractors is twofold. First, nearly optimal and very efficient constructions of seeded extractors are known in the literature that extract nearly the entire source entropy
with only a short seed. This allows us to attain nearly optimal rates for the non-adaptive case. Furthermore, and crucially, such nearly optimal extractor constructions (in particular, Trevisan's extractor \cite{Trevisan,improvement}) can in fact be linear functions for every fixed choice of the seed (in contrast, seedless affine extractors can never be linear functions). We take advantage of the linearity of the extractor in a crucial way and use a rather delicate analysis to show that in fact the linearity of the extractor can be utilized to prove that the resulting secret sharing scheme provides the stringent worst-case secret guarantee. 
The construction and its proof follows similarly as the optimal construction of binary non-adaptive secret sharing in \cite{preprint}. 

Our construction of adaptive leakage-resilient secret sharing uses a classical paradigm of combining a seedless extractor with a seeded extractor \cite{moremileage}. The seedless extractor extracts a short seed for the seeded extractor and the combination is in effect a seedless extractor with the good properties of the seeded extractor. 
Unfortunately, the asymptotic optimal property that is enjoyed by the seeded extractor based construction above is not preserved due to an error bounding process that also involves the affine extractor. It is an interesting open problem that whether this slack can be tighten. Nevertheless, using this classical paradigm gives us significantly better parameters, which also has independent interest in secret sharing over small constant share size \cite{ref:CCX13,BGK16,preprint}. We in fact obtain, as a side result, an explicit secret sharing against an adaptive adversary with significantly  better parameters than \cite{preprint} (see {\bf Related works}). The improvement comes from making good use of the linearity of the seeded extractor and a more efficient way of inverting the extractor that exploits this classical structure.

 \bigskip
 \noindent
{\em Affine leakage-resilient non-malleable secret sharing.}
In a nut shell, our constructions start with the extractor based construction of secret sharing scheme and strengthen the extractor towards obtaining non-malleability. 
This idea is inspired by the following extractor based construction of non-malleable codes.
An important theoretical discovery 
 in constructions of non-malleable codes is the connection between non-malleable codes and invertible {\em seedless non-malleable extractors} by Cheraghchi and Guruswami \cite{ChGu1}. A seedless non-malleable extractor is defined with respect to a family of tampering functions, 
which are applied to the input of the extractor. Non-malleability here means that the output corresponding to the original input is independent of the output of a tampered input. Intuitively, if one uses the extractor as the decoder then non-malleability of the obtained code follows naturally from the independence of the two outcomes. This connection plays an important role in the construction of $C$-split state non-malleable codes \cite{10split,CGL16,Li17,XinLi18,CL19}. This result was recently extended to affine tampering functions through explicitly constructing seedless non-malleable extractors with respect to affine tampering functions \cite{CL17}.

Our construction of adaptive affine leakage-resilient non-malleable secret sharing with respect to $\mathcal{F}_\mathsf{affine}$ strengthens the affine extractor to an affine non-malleable extractor. 
Intuitively, we trivially have an affine leakage-resilient secret sharing scheme, since an affine non-malleable extractor is in particular an affine extractor. We can further show that the scheme is non-malleable
The analysis is again focused on the uniform pre-image (of a uniform secret) generated by the inverter of the extractor. As argued before, conditioned on a view $\mathsf{v}$ of the $t$ shares and $\beta$-bounded affine leakage adversary $\mathcal{A}$, the uniform pre-image becomes an affine source. Under the same conditioning, the affine tampering strategy $\sigma$ outputs the corresponding affine tampering function $f^\mathsf{v}$ that is applied to the share vector. Due to the linearity of the erasure correcting code, this $f^\mathsf{v}$ induces an affine tampering function $g$ that is in effect applied to the pre-image we are investigating. If the extractor can non-malleably (with respect to affine functions) extract from the affine source, the tampered outcome is independent of the original secret.  
We then obtain a clean reduction from affine leakage-resilient non-malleable secret sharing to affine non-malleable extractors (see Theorem \ref{th: ANME}).

Our construction of adaptive affine leakage-resilient non-malleable secret sharing with respect to $\mathcal{F}_{\mathsf{BIT}}$ is built on the particular adaptive affine leakage-resilient secret sharing construction above. We strengthen the linear seeded extractor to a linear {\em seeded} non-malleable extractor. Seeded non-malleable extractors were proposed (in fact before the notion of seedless non-malleable extractors) for application in privacy amplification over public unauthenticated discussion~\cite{DW07}. A seeded non-malleable extractor is very different from its seedless counterpart and the only thing that these two objects have in common is to achieve independence of the original extractor output from the tampered extractor output. The first difference lies in what is tampered. The source of the seeded extractor is not tampered, it is its seed that is tampered. The second difference lies in what tampering is allowed. The seed tampering of the seeded extractor is not restricted by a family of functions, but is allowed to be any tampering function as long as it does not have any fixed points. 
We overcome the first difference through suitably conditioning on some event such that the tampered source is equal to the original source adding a constant offset, thanks to restriction to $\mathcal{F}_{\mathsf{BIT}}\subset\mathcal{F}_{\mathsf{affine}}$. Since the seeded non-malleable extractor is linear, we can separate the constant offset from the tampered source completely and reduce to the same source situation. We overcome the second difference through detecting the tampering, whenever the tampered seed coincides the original seed,  using an Algebraic Manipulation Detection AMD \cite{CDFP08} pre-coding of the secret. We can not guarantee that the tampered share vector always leads to a seed different from the original seed. But when the two seeds do coincide, as mentioned a few lines ago, the linearity of the non-malleable extractor allows for separating out an additive offset. This results in reconstructing an (obliviously) additively tampered secret, which is easily detected using, for example, the AMD code \cite{CDFP08}.

\remove{
Given a {\em linear} seeded non-malleable extractor that can extract from a constant fraction of entropy, we do obtain non-malleable secret sharing 
in adaptive tampering model with constant relative thresholds $0<\tau<\rho<1$ (see Theorem \ref{th: high}). 
Seeded non-malleable extractors is well studied topic and there are many good constructions. But when restricted to linear case, as far as we know, there is only the inner product based construction \cite{XinLiIP}, which requires entropy rate around half.This entropy rate around half barrier existed in the literature of (non-linear) non-malleable extractor constructions \cite{DLTWZ14}, but was quickly overcome \cite{overcome}. 
We prove the existence of the required linear seeded non-malleable extractors using a probabilistic argument (see Theorem \ref{th: existence}) and leave the explicit construction as an interesting open question.
 }
 

\smallskip
\noindent
{\bf Related works.} 

Another line of works related to the current work is the study of ramp secret sharing over a constant share size $q$. 
The main characteristics of this line of works are fixed share size $q$, unconstrained number $N$ of players and ramp parameters $(t,r)$ satisfying $t=\tau N$, $r=\rho N$. The goal is minimizing the relative threshold gap $\gamma=\frac{g}{N}=\frac{r-t}{N}=\rho-\tau$. It is shown in \cite{ref:CCX13} and~\cite{BGK16} that for 
$0< t< N-1$,
$$
g\geq(N+2)/(2q-1).
$$
This means that once $q$ is fixed, the relative gap $\gamma=\frac{g}{N}>\frac{1}{2q-1}$. In particular, when $q=2$, we must have $\gamma>1/3$.
This constraint is recently showed avoidable once the perfect privacy and perfect reconstruction of the ramp secret sharing are relaxed to statistical privacy (any $t$ shares from a pair of secrets have a statistical distance negligible in $N$) and probabilistic reconstruction (reconstruction with $r$ shares has a failure probability that is negligible in $N$), respectively \cite{preprint}. It is shown that for any $0\leq\tau<\rho\leq1$, ramp secret sharing families (with relaxed privacy and reconstruction) can be explicitly constructed such that the privacy threshold $t=\tau N$ and the reconstruction threshold $r=\rho N$. The non-perfect privacy brings out the distinction between an adaptive reading adversary and a non-adaptive reading adversary. The authors then give two constructions for these two types of reading adversaries, respectively. In particular, the construction for non-adaptive adversary shares a secret of $N(\rho-\tau-o(1))\log q$ bits, which they show is optimal. The construction for adaptive adversary does not achieve this secret length and the authors leave improving the secret length as an open problem.
As mentioned previously, the tools developed in our second construction of affine leakage resilient secret sharing can be used to significantly improve the secret length of the construction in \cite{preprint}. In particular, a {\em coding rate} ($\ell/N$) was used as the design criterion and shown to be upper-bounded by $\rho-\tau$. The coding rate is related to the information ratio in the current work as follows. 
$$
\mbox{Information ratio}=\frac{N/P}{\ell} =\frac{1/P}{\mbox{coding rate}}.
$$

\bigskip

The rest of the paper is organised as follows. Section \ref{sec: pre} contains the definitions of various randomness extractors that appear in this work. Section \ref{sec: model} contains two constructions of affine leakage-resilient secret sharing, for non-adaptive adversary and adaptive adversary, respectively. Section \ref{sec: reduction} contains two constructions of adaptive affine leakage-resilient secret sharing that are non-malleable with respect to affine tampering and bit-wise independent tampering, respectively.

\section{Preliminaries}\label{sec: pre}
Coding schemes define the basic properties  for codes (schemes) that are used in cryptography. Let  $\bot$  denote  a special symbol that means detection.  

\begin{definition}[\cite{DzPiWi}] A $(k,n)$-coding scheme consists of two polynomial-time functions: a randomised 
encoding function 
$\mathsf{Enc}:\{0,1\}^k\rightarrow\{0,1\}^n$, where the randomness is implicit, and a deterministic decoding function $\mathsf{Dec}:\{0,1\}^n\rightarrow\{0,1\}^k\cup\{\perp\}$ such that, for each $\mathsf{m}\in\{0,1\}^k$, $\mathsf{Pr}[\mathsf{Dec}(\mathsf{Enc}(\mathsf{m}))=\mathsf{m}]=1$ (correctness), and the probability is  over the randomness of the encoding algorithm.
\end{definition}

The \emph{statistical distance} of two random variables (their corresponding distributions) is defined as follows. For $\mathsf{X},\mathsf{Y}\leftarrow\Omega$, 
$$\mathsf{SD}(\mathsf{X};\mathsf{Y})= \dfrac{1}{2}\sum_{\mathbf{\omega} \in \Omega}|\mbox{Pr}(\mathsf{X}=\mathbf{\omega})-\mbox{Pr}(\mathsf{Y}=\mathbf{\omega})|.$$
We say $\mathsf{X}$ and $\mathsf{Y}$ are $\varepsilon$-close (denoted $\mathsf{X}\stackrel{\varepsilon}{\sim}\mathsf{Y}$) if $\mathsf{SD}(\mathsf{X},\mathsf{Y})\leq \varepsilon$.

A {\em tampering function}  
for a $(k,n)$-coding scheme is a function $f:\{0,1\}^n\rightarrow\{0,1\}^n$.
\begin{definition}[\cite{DzPiWi}] \label{def: non-malleability}Let $\mathcal{F}$ be a family of tampering functions.
 For each $f\in\mathcal{F}$ and $\mathsf{m}\in\{0,1\}^k$, define the tampering-experiment
$$
\mathrm{Tamper}_\mathsf{m}^f=\left\{
\begin{array}{c}
\mathbf{x}\leftarrow\mathsf{Enc}(\mathsf{m}),\tilde{\mathbf{x}}= f(\mathbf{x}),\tilde{\mathsf{m}}=\mathsf{Dec}(\tilde{\mathbf{x}})\\
\mathrm{Output}\ \tilde{\mathsf{m}},\\
\end{array}
\right\}.
$$
which is a random variable over the randomness of the encoding function  $\mathsf{Enc}$.
A coding scheme $(\mathsf{Enc},\mathsf{Dec})$ is  {\em non-malleable with respect to $\mathcal{F}$}  if for each $f\in\mathcal{F}$, there exists a distribution
 $\mathcal{D}_f$ over the set $\{0,1\}^k\bigcup\{\perp, \mathsf{same}^*\}$, such that, for all $\mathsf{m}\in\{0,1\}^k$, we have:
\begin{equation}\label{eq: NMCdef}
\mathrm{Tamper}_\mathsf{m}^f\stackrel{\varepsilon}{\sim}\left\{
\begin{array}{c}
\tilde{\mathsf{m}}\leftarrow\mathcal{D}_f\\
\mathrm{Output}\ \mathsf{m}\ \mathrm{ if }\ \tilde{\mathsf{m}}=\mathsf{same}^*,\ \mathrm{and}\ \tilde{\mathsf{m}}\ \mathrm{ otherwise;}
\end{array}
\right\}
\end{equation}
and $\mathcal{D}_f$ is  efficiently samplable given oracle access to $f(\cdot)$. 
\end{definition}
The right hand side of (\ref{eq: NMCdef}) is sometimes denoted  by $\mathsf{Copy}(\mathcal{D}_f,\mathsf{m})$. 
Using this notation,  (\ref{eq: NMCdef}) can be written as,
\begin{equation} \tag{\ref{eq: NMCdef}'}
\mathrm{Tamper}_\mathsf{m}^f\stackrel{\varepsilon}{\sim}\mathsf{Copy}(\mathcal{D}_f,\mathsf{m}). 
\end{equation}

The following coding scheme, originally proposed for constructing robust secret sharing, is frequently used as a building block for constructing non-malleable codes.
\begin{definition}[\cite{CDFP08}]\label{def: AMD} Let $(\mathsf{AMDenc}, \mathsf{AMDdec})$ be a coding scheme with $\mathsf{AMDenc}: \{0,1\}^k\rightarrow\{0,1\}^n$. We say that $(\mathsf{AMDenc}, \mathsf{AMDdec})$ is a $\delta$-secure Algebraic Manipulation Detection (AMD) code if for all $\mathsf{m}\in\{0,1\}^k$ and all non-zero $\Delta\in\{0,1\}^n$, we have $\mathsf{Pr}[\mathsf{AMDdec}(\mathsf{AMDenc}(\mathsf{m})+\Delta)\notin\{\mathsf{m},\bot\}]\leq \delta$, where the probability is over the randomness of the encoding.
\end{definition}



An explicit optimal construction of AMD code is given in \cite{CDFP08} that in  fact gives a \textit{tamper detection code} \cite{TDC}. 
We say 
an AMD code 
achieves \textit{$\delta$-tamper detection security}  if for all $\Delta\neq 0^n$, $\mathsf{Pr}[\mathsf{AMDdec}(\mathsf{AMDenc}(\mathsf{m})+\Delta)\neq\bot]\leq \delta$.
%

We use various types of randomness extractors in our constructions. Randomness extractors extract close to uniform bits from input sequences that are not uniform but have some guaranteed entropy. See \cite{nisan-1996} and references there in for more information about randomness extractors. 

A {\em randomness source} is a random variable with  lower bound on its min-entropy, which is defined by $\mathsf{H}_\infty(\mathsf{X})=-\log \max_\mathsf{x}\{\mathsf{Pr}[\mathsf{X}=\mathsf{x}]\}$.
We say  a  random variable $\mathsf{X} \leftarrow \lbrace 0,1\rbrace^{n}$ is a {\em $(n,k)$-source},  if $\mathsf{H}_\infty(\mathsf{X})\geq k$. 
For well structured sources, there exist deterministic functions that can extract close to uniform bits.
An  \textit{affine} $(n,k)$-source is a random variable that is uniformly distributed on an affine translation of some $k$-dimensional sub-space of $\lbrace0,1\rbrace^n$. Let $\mathsf{U}_m$ denote the random variable uniformly distributed over $\{0,1\}^m$.


\begin{definition}\label{def: AExt} A function $\mathsf{aExt}\colon\lbrace0,1\rbrace^n \to \lbrace 0,1 \rbrace ^m$ is an affine $(k, \varepsilon)$-extractor if for any affine $(n,k)$-source $\mathsf{X}$, we have 
$$
\mathsf{SD}(\mathsf{aExt}(\mathsf{X});\mathsf{U}_m) \leq \varepsilon.
$$ 
\end{definition}



We will use Bourgain's affine extractor (or the alternative \cite{XinLiAffine} due to Li) in our constructions. 
\begin{lemma}[\cite{Bourgain}]\label{lem: affine} 
For every constant $0<\mu\leq 1$, there is an explicit affine extractor $\mathsf{aExt}\colon \{0,1\}^n\rightarrow\{0,1\}^m$ for affine $(n,n\mu)$-sources with output length $m=\Omega(n)$ and error at most $2^{-\Omega(n)}$.
\end{lemma}

For general $(n,k)$-sources, there does not exist a deterministic function that can extract close to uniform bits from all of them simultaneously. A family of deterministic functions are needed.

\begin{definition}\label{def: strong extractor}
A function $\mathsf{Ext}\colon\lbrace0,1\rbrace^d\times \lbrace0,1\rbrace^n  \to \lbrace 0,1 \rbrace ^m$ is a strong seeded $(k, \varepsilon)$-extractor if for any $(n,k)$-source $\mathsf{X}$, we have 
$$
\mathsf{SD}(\mathsf{S},\mathsf{Ext}(\mathsf{S},\mathsf{X});\mathsf{S},\mathsf{U}_m) \leq \varepsilon,
$$ 
where $\mathsf{S}$ is chosen uniformly from $\lbrace0,1\rbrace^d$. 
A seeded extractor $\mathsf{Ext}(\cdot,\cdot)$ is called linear if for any fixed seed $\mathsf{S}=\mathsf{s}$, the function $\mathsf{Ext}(\mathsf{s},\cdot)$ is a linear function.
\end{definition}

There are linear seeded extractors that extract all the randomness, for example, the Trevisan's extractor \cite{Trevisan}. In particular, we use the following improvement of this extractor due to Raz, Reingold and Vadhan \cite{improvement}.


\begin{lemma}[\cite{improvement}] \label{lem: Ext}
There is an explicit linear strong $(k,\varepsilon)$-extractor $\mathsf{Ext}\colon\lbrace0,1\rbrace^d\times \lbrace0,1\rbrace^n  \to \lbrace 0,1 \rbrace ^m$ with $d=O(\log^3(n/\varepsilon))$ and $m=k-O(d)$.
\end{lemma}

Non-malleability of randomness extractors captures their tolerance against tampering. It was first defined for seeded extractors by Dodis and Wichs \cite{DW07} with application in privacy amplification over public and unauthenticated discussion. The tampering considered is an arbitrary seed tampering that does not have any fixed point.

\begin{definition}[\cite{DW07}] \label{def: NMExt}
A seeded $(k,\varepsilon)$-non-malleable 
extractor is a function $\mathsf{nmExt}: \{0,1\}^d\times\{0,1\}^n\rightarrow\{0,1\}^m$ such that given any 
$(n,k)$-source $\mathsf{X}$, an independent uniform seed $\mathsf{Z}\in\{0,1\}^d$, for any (deterministic) function $\mathcal{A}:\{0,1\}^d\rightarrow\{0,1\}^d$ such that $\mathcal{A}(\mathsf{z})\neq\mathsf{z}$ for any $\mathsf{z}$, we have  
\begin{equation}\label{eq: NMExt}
\mathsf{SD}(\mathsf{Z},\mathsf{nmExt}(\mathcal{A}(\mathsf{Z}),\mathsf{X}),\mathsf{nmExt}(\mathsf{Z},\mathsf{X});\mathsf{Z},\mathsf{nmExt}(\mathcal{A}(\mathsf{Z}),\mathsf{X}),\mathsf{U_m})\leq \varepsilon.
\end{equation}
\end{definition}

Non-malleable seedless extractors were proposed by Cheraghchi and Guruswami for constructing non-malleable codes. The tampering now is a source tampering and is restricted to a particular tampering family.

\begin{definition}[\cite{ChGu1}] \label{def: nmExt}
A function $\mathsf{nmExt}\colon\{0,1\}^n\rightarrow\{0,1\}^m$ is a $(k,\varepsilon)$-seedless non-malleable extractor with respect to a class $\mathcal{X}$ of sources over $\{0,1\}^n$ and a class $\mathcal{F}$ of tampering functions acting on $\{0,1\}^n$, if for every $\mathsf{X}\in\mathcal{X}$ with min-entropy $k$ and every $f\in\mathcal{F}$, there is a distribution $\mathcal{D}_f$ over $\{0,1\}^m\cup\{\mathsf{same}^{*}\}$ such that for an independent $\mathsf{Y}$ sampled from $\mathcal{D}_f$, we have
\begin{equation}\label{eq: nmExt}
\mathsf{SD}(\mathsf{nmExt}(f(\mathsf{X}),\mathsf{nmExt}(\mathsf{X});\mathsf{Copy}(\mathsf{Y},\mathsf{U}_m),\mathsf{U_m})\leq \varepsilon,
\end{equation}
where the two copies of $\mathsf{U}_m$ denote the same random variable and $\mathsf{Copy}(\mathsf{y},\mathsf{u})=\mathsf{y}$ always except when $\mathsf{y}=\mathsf{same}^{*}$, in which case it outputs $\mathsf{u}$.
\end{definition}

We will use Chattopadhyay and Li's affine {\em non-malleable} extractor. 
We first give the restricted form of the extractor, where the source tampering function does not have any fixed points.

\begin{lemma}[\cite{CL17}]\label{lem: ANME}
For all $n,k>0$, any $\delta>0$ and $k\geq n-n^{\frac{\delta}{2}}$, there exists an efficient function $\mathsf{anmExt}\colon\{0,1\}^n\rightarrow\{0,1\}^m$, $m=n^{\Omega(1)}$, such that if $\mathsf{X}$ is an affine $(n,k)$-source and $\mathcal{A}\colon\{0,1\}^n\rightarrow\{0,1\}^n$ is an affine function with no fixed point, then 
$$
\mathsf{SD}(\mathsf{anmExt}(\mathcal{A}(\mathsf{X})),\mathsf{anmExt}(\mathsf{X});\mathsf{anmExt}(\mathcal{A}(\mathsf{X})),\mathsf{U}_m)\leq2^{-n^{\Omega(1)}}.
$$
\end{lemma}

Let $\mathcal{F}_{\mathsf{affine}}$ be the set of tampering functions from $\{0,1\}^n$ to $\{0,1\}^n$ where each output bit is an affine function of the input bits. The affine non-malleable extractors in Lemma \ref{lem: ANME} can be easiliy converted into a seedless non-malleable extractor with respect to $\mathcal{F}_{\mathsf{affine}}$.

\begin{lemma}[\cite{CL17}]\label{lem: convert}
Let $\mathsf{anmExt}\colon\{0,1\}^n\rightarrow\{0,1\}^m$ be a $(k-\eta,\varepsilon)$-non-malleable extractor for affine sources, with respect to affine tampering functions with on fixed points. Then $\mathsf{anmExt}$ is a $(k,\varepsilon+(n+1)2^{-\eta})$-non-malleable extractor for affine sources, with respect to $\mathcal{F}_{\mathsf{affine}}$.
\end{lemma}

Explicit constructions of randomness extractors have efficient forward direction of extraction. In some applications, we usually need to efficiently invert the process: Given an extractor output, sample a random pre-image. 
This is not necessarily efficient if the extractor is not a linear function, in which case we need to explicitly construct an {\em invertible extractor}.
If the extractor is linear, sampling a random pre-image can be done in polynomial time. In general,

\begin{definition}[\cite{CDS12}]\label{def: invertible} 
Let $f$ be a mapping from $\{0,1\}^n$ to $\{0,1\}^m$. For $v\geq0$, a function $\mathsf{Inv}\colon \{0,1\}^m\times\{0,1\}^r \rightarrow\{0,1\}^n$ is called a $v$-inverter for $f$ if the following conditions hold:
\begin{itemize}
\item (Inversion) Given $\mathsf{y}\in\{0,1\}^m$ such that its pre-image $f^{-1}(\mathsf{y})$ is nonempty, for every $\mathsf{r}\in\{0,1\}^r$ we have $f (\mathsf{Inv}(\mathsf{y}, \mathsf{r})) = \mathsf{y}$.
\item (Uniformity) $\mathsf{Inv}(\mathsf{U}_m,\mathsf{U}_r)$ is $\mu$-close to $\mathsf{U}_n$.
\end{itemize}
A $\mu$-inverter is called efficient if there is a randomized algorithm that runs in worst-case polynomial time and, given $\mathsf{y}\in\{0,1\}^m$ and $\mathsf{r}$ as a random seed, computes $\mathsf{Inv}(\mathsf{y}, \mathsf{r})$. We call a mapping $\mu$-invertible if it has an efficient $\mu$-inverter, and drop the prefix $\mu$ from the notation when it is zero. We abuse the notation and denote the inverter of $f$ by $f^{-1}$.
\end{definition}

Finally, we need the following simple lemma whose proof can be found in Appendix \ref{apdx: conditioning}.

\begin{lemma}\label{lem: conditioning} Let $\mathsf{V},\mathsf{V}'$ be two random variables distributed over the set $\mathcal{V}$ and $\mathsf{W},\mathsf{W}'$ over $\mathcal{W}$ satisfying $\mathsf{SD}(\mathsf{V},\mathsf{W};\mathsf{V}',\mathsf{W}')\leq \varepsilon$. Let $\mathcal{E}\subset\mathcal{W}$ be an event. Then we have the following.
$$
\mathsf{SD}(\mathsf{V}|\mathsf{W}\in\mathcal{E};\mathsf{V}'|\mathsf{W}'\in\mathcal{E})\leq \frac{2\varepsilon}{\mathsf{Pr}[\mathsf{W}'\in\mathcal{E}]}.
$$
\end{lemma}

\smallskip
\noindent
A {\em stochastic code} 
has a randomised encoder and a deterministic decoder. The encoder $\mathsf{Enc}\colon\{0,1\}^m\times\mathcal{R}\rightarrow \{0,1\}^n$ uses local randomness $\mathsf{R}\leftarrow\mathcal{R}$ to encode a message $\mathsf{m}\in\{0,1\}^m$. The decoder is a deterministic function $\mathsf{Dec}\colon\{0,1\}^n\rightarrow \{0,1\}^m\cup\{\bot\}$. The decoding probability is defined over the encoding randomness $\mathsf{R}\leftarrow\mathcal{R}$. 
Stochastic codes are known to explicitly achieve the capacity of some special adversarial channels \cite{GS10}.

Affine sources play an important role in our constructions. We define a general requirement for the stochastic code used in our constructions. 
\begin{definition}[\cite{preprint}]\label{def: affine} Let $\mathsf{Enc}\colon\{0,1\}^m\times\mathcal{R}\rightarrow \{0,1\}^n$ be the encoder of a stochastic code. We say it is a stochastic affine code if for any $\mathsf{r}\in\mathcal{R}$, the encoding function $\mathsf{Enc}(\cdot,\mathsf{r})$ specified by $\mathsf{r}$ is an affine function  
of the message. That is we have 
\begin{align*}
\mathsf{Enc}(\mathsf{m},\mathsf{r})=\mathsf{m}G_\mathsf{r}+\Delta_\mathsf{r},
\end{align*}
where $G_\mathsf{r}\in\{0,1\}^{m\times n}$ 
and $\Delta_\mathsf{r}\in\{0,1\}^n$ are specified by 
the randomness $\mathsf{r}$.

\end{definition}

We then adapt a construction in \cite{GS10} to obtain the following capacity-achieving Stochastic Affine-Erasure Correcting Code (SA-ECC). 
In particular, we show for any $p\in[0,1)$, there is an explicit stochastic affine code that corrects $p$ fraction of adversarial erasures and achieves the rate $1-p$. 


\begin{lemma}[\cite{preprint}]
\label{th: SA-ECC} 
For every $p\in[0,1)$, and every $\xi>0$, there is an efficiently encodable and decodable stochastic affine code $(\mathsf{Enc},\mathsf{Dec})$ with rate $R=1-p-\xi$ such that for every $\mathsf{m}\in\{0,1\}^{NR}$ and erasure pattern of at most $p$ fraction, we have $\mathsf{Pr}[\mathsf{Dec}(\widetilde{\mathsf{Enc}(\mathsf{m})})=\mathsf{m}]\geq 1-\exp(-\Omega(\xi^2N/\log^2 N))$, where $\widetilde{\mathsf{Enc}(\mathsf{m})}$ denotes the partially erased random codeword and $N$ denotes the length of the codeword.
\end{lemma}


\section{Affine Leakage-Resilient Secret Sharing}\label{sec: model} 
In this section, we study a new leakage model for secret sharing. All results are stated as $(t,r,P)$-ramp schemes. The special results concerning $r$-out-of-$P$ threshold schemes can be recovered through letting $t=r-1$. We start with recalling the  {\em Leakage-Resilient Storage (LRS)} model of \cite{DDV10}.


A leakage-resilient storage scheme is a pair $(\mathsf{Enc},\mathsf{Dec})$, where $\mathsf{Enc}\colon\{0,1\}^\ell\times\mathcal{R}\rightarrow\{0,1\}^N$ is a randomised, efficiently computable function ($\mathcal{R}$ is the randomness set) and $\mathsf{Enc}\colon\{0,1\}^\ell\rightarrow\{0,1\}^N$ is a deterministic, efficiently computable function. 
Consider the following game between an adversary $\mathcal{A}$ and an oracle $\mathcal{O}$.
\begin{enumerate}
\item The adversary $\mathcal{A}$ chooses a pair of messages $\mathsf{m}_0,\mathsf{m}_1\in\{0,1\}^\ell$ and sends them to the oracle $\mathcal{O}$.
\item The oracle $\mathcal{O}$ chooses a random bit $b\in\{0,1\}$ and compute $\mathsf{Enc}(\mathsf{m}_b)$.
\item The following is executed $\theta$ times, for $i=1,\ldots,\theta$:
     \begin{enumerate}
     \item $\mathcal{A}$ selects a function $l_i\colon\{0,1\}^N\rightarrow\{0,1\}^{c_i}$ from a set $\mathcal{L}$ of functions, and sends it to $\mathcal{O}$,
     \item $\mathcal{O}$ sends $l_i(\mathsf{Enc}(\mathsf{m}_b))$ to $\mathcal{A}$. This is called $\mathcal{A}$ retrieves $c_i$ bits through $\mathcal{L}$-leakage.
     \end{enumerate}
\end{enumerate}
We will call the adversary $\mathcal{A}$ a $\beta$-bounded $\mathcal{L}$-leakage adversary if $\sum_{i=1}^\theta c_i\leq \beta$. 

We consider statistical secret sharing, where the privacy with respect to a given access structure is defined using indistinguishability of unauthorised set of shares for a pair of secrets. The privacy adversary's choice of shares may be non-adaptive or adaptive, which become different notions when the privacy error is non-zero. We want to consider leakage-resiliency for statistical secret sharing on top of the privacy with respect to a given access structure. We now view the full share vector of the secret sharing for $P$ players as an encoding of a secret $\mathsf{s}\in\{0,1\}^\ell$ (the sharing algorithm being the randomised encoder) in the codeword space $\{0,1\}^N$, where $N=P\log q$ and $q$ is the share size. A non-compartmentalized leakage model means that the set $\mathcal{L}$ contains leakage functions whose outputs can depend on all parts of the full share vector. Inspired by the non-compartmentalized tampering models considered in the non-malleable codes literature, we study the set $\mathcal{L}_\mathsf{affine}$ of $\mathbb{F}_2$-affine leakage functions. Each output bit of a $\mathbb{F}_2$-affine leakage function $l\colon\{0,1\}^N\rightarrow\{0,1\}^{c}$ is an affine function of the input in $\{0,1\}^N$. 


\begin{definition} \label{def: lrSSS} For integers $0\leq t<r\leq P$, a $(\varepsilon(N),\delta(N))$-statistical secret sharing for ramp parameters $(t,r,P)$ that is leakage-resilient against a $\beta$-bounded $\mathcal{L}_\mathsf{affine}$-leakage adversary 
is a pair of polynomial-time algorithms $(\mathsf{Share},\mathsf{Recst})$,
$$
\mathsf{Share}\colon\{0,1\}^{\ell(N)}\times\mathcal{R}\rightarrow\{0,1\}^N,\ P|N
$$
where $\mathcal{R}$ denotes the randomness set, and for any reconstruction set $R\subset \mathcal{P}$ of size $|R|=r$, 
$$
\mathsf{Recst}_R\colon\left(\{0,1\}^{N/P}\right)^r\rightarrow\{0,1\}^{\ell(N)}\cup\{\bot\},
$$ 
that satisfy the following properties.\begin{itemize}

\item Correctness: Given any $r$ out of the $P$ blocks of the share vector $\mathsf{Share}(\mathsf{s})$, the reconstruct algorithm $\mathsf{Recst}$ reconstructs the secret $\mathsf{s}$ with probability at least $1-\delta(N)$. 


\item Privacy and leakage-resiliency: 

 \begin{itemize}
    \item Non-adaptive adversary: for any pair $\mathsf{s}_0,\mathsf{s}_1\in\{0,1\}^{\ell(N)}$ of secrets, any $A\subset \mathcal{P}$ of size $|A|\leq t$, 
    any affine leakage function $l\colon\{0,1\}^N\rightarrow\{0,1\}^{c}$ with $c\leq\beta$,      \begin{equation}\label{eq: WtII security}
      \mathsf{SD}(l(\mathsf{Share}(\mathsf{s}_0)),\mathsf{Share}(\mathsf{s}_0)_{A};l(\mathsf{Share}(\mathsf{s}_1)),\mathsf{Share}(\mathsf{s}_1)_{A})\leq \varepsilon(N),
     \end{equation}
     where $\mathsf{Share}(\mathsf{s})_{A}$ denotes the projection of $\mathsf{Share}(\mathsf{s})\in\left(\{0,1\}^{N/P}\right)^P$ on the blocks specified by $A$.

    \item Adaptive adversary: 
    For any $\mathsf{s}_0,\mathsf{s}_1\in \{0,1\}^{\ell(N)}$ and any adaptive adversary $\mathcal{A}_{\beta,\mathsf{affine}}$ that is $\beta$-bounded and affine, 
    \begin{equation}\label{eq: adaptive privacy}
    \mathsf{SD}\left (\mathsf{View}_{\mathcal{A}_{\beta,\mathsf{affine}}}^{\mathcal{O}(\mathsf{Share}(\mathsf{s}_0))};
    \mathsf{View}_{\mathcal{A}_{\beta,\mathsf{affine}}}^{\mathcal{O}(\mathsf{Share}(\mathsf{s}_1))}\right )\leq \varepsilon(N),
    \end{equation}
    where $\mathsf{View}_{\mathcal{A}_{\beta,\mathsf{affine}}}^{\mathcal{O}(\mathsf{Share}(\mathsf{s}))}$ denotes the view of the adversary $\mathcal{A}_{\beta,\mathsf{affine}}$ after playing the LRS game described above with the oracle $\mathcal{O}$ and, at any time, can adaptively select up to $t$ shares to append to the messages retrieved from $\mathcal{O}(\mathsf{Share}(\mathsf{s}))$.
   \end{itemize}
   

\end{itemize}
When it is clear from the context, instead of $\varepsilon(N), \delta(N), \ell(N)$, we write $\varepsilon,\delta,\ell$. 
\end{definition}

In the sequel, we simply refer to the objects defined in Definition \ref{def: lrSSS} non-adaptive/adaptive affine leakage-resilient secret sharing.

\subsection{Non-adaptive Affine Leakage-Resilient Secret Sharing}
We first give a construction of non-adaptive affine leakage-resilient secret sharing.

\begin{theorem}\label{th: non-adaptive} 
Let $\mathsf{Ext} \colon \{0,1\}^d\times\{0,1\}^n\rightarrow\{0,1\}^\ell$ be a linear strong seeded $(n-\tau N-\beta,\frac{\varepsilon}{8})$-extractor and $\mathsf{Ext^{-1}}(\mathsf{z},\cdot) \colon \{0,1\}^\ell\times\mathcal{R}_1\rightarrow\{0,1\}^n$ be the inverter of the function $\mathsf{Ext}(\mathsf{z},\cdot)$ that maps an $\mathsf{s}\in\{0,1\}^\ell$ to one of its pre-images chosen uniformly at random. 
Let $(\mathsf{SA\mbox{-}ECCenc},\mathsf{SA\mbox{-}ECCdec})$ be a stochastic affine-erasure correcting code with the encoder $\mathsf{SA\mbox{-}ECCenc}\colon\{0,1\}^{d+n}\times\mathcal{R}_2\rightarrow\{0,1\}^N$ that tolerates $N-\rho N$ bit erasures and decodes with success probability at least $1-\delta$.   
Then the following coding scheme $(\mathsf{Share},\mathsf{Recst})$ is a non-adaptive affine leakage-resilient secret sharing with security parameters $\varepsilon$, $\delta$, leakage bound $\beta$ and ramp parameters $(t,r,P)$ such that $\tau=t/P$, $\rho=r/P$. 
\begin{align*}
\left\{
\begin{array}{ll}
\mathsf{Share}(\mathsf{s})&=\mathsf{SA\mbox{-}ECCenc(Z||Ext^{-1}}(\mathsf{Z},\mathsf{s})),\mathrm{where}\  \mathsf{Z}\stackrel{\$}{\leftarrow}\{0,1\}^d;\\
\mathsf{Recst}(\tilde{\mathsf{y}})&=\mathsf{Ext(\mathsf{z},\mathsf{x}), \mathrm{where}\  (\mathsf{z}||\mathsf{x})=SA\mbox{-}ECCdec}(\tilde{\mathsf{y}}).
\end{array}
\right.
\end{align*}
Here $\tilde{\mathsf{y}}$ denotes an incomplete version of a share vector $\mathsf{y}\in\{0,1\}^N$ with some of its components replaced by erasure symbols.
\end{theorem}

The proof is similar to the proof for the optimal construction of non-adaptive binary ramp scheme in \cite{preprint} and is given in Appendix B for completeness. We provide the intuition  of the construction here, starting with the high-level idea with an affine extractor $\mathsf{aExt}$, which is shared by our new construction for adaptive adversary in the next subsection. 

Intuitively, the $\mathsf{SA\mbox{-}ECC}$ enables the reconstruction from any $\rho N$ bits. The privacy for any $\tau N$ shares is not as straightforward. Imagine we share a uniformly distributed random secret $\mathsf{S}\stackrel{\$}{\leftarrow}\{0,1\}^\ell$ and want to find out the distribution of the secret conditioned on the adversary's view $\mathsf{V}$, which is consist of up to $t$ shares and up to $\beta$ bits retrieved through applying an affine leakage function. Intuitively, if the distribution of the uniform secret conditioned on $\mathsf{V}$ remains uniform, we have privacy and leakage-resiliency. Since we are using extractors to extract uniform distribution, the focus is then to make sure the source has enough entropy and is of the right structure (affine source). 
According to the definition of an inverter, if the secret has uniform distribution $\mathsf{U}_\ell$, then the inverter outputs a uniform distribution $\mathsf{U}_n$. 
$$
\mathsf{U}_n\stackrel{\mu}{\sim}\mathsf{aExt}^{-1}(\mathsf{U}_\ell).
$$
On the other hand, in the construction, the source is the message of the $\mathsf{SA\mbox{-}ECC}$, which is an affine function. Obtaining any $\tau N$ bits of the $\mathsf{SA\mbox{-}ECC}$ codeword is equivalent to applying an affine function of $\tau N$-bit output to the source. Moreover, applying an affine leakage function $l\colon\{0,1\}^N\rightarrow\{0,1\}^\beta$ to the $\mathsf{SA\mbox{-}ECC}$ codeword is equivalent to applying the composition $l\circ\mathsf{SA\mbox{-}ECC}$ to the source.
An affine function induces a partition of the space $\{0,1\}^n$ into cosets each corresponding to a particular value of the adversary's view $\mathsf{V}=\mathsf{v}$. 
Given that the adversary observes $\mathsf{V}=\mathsf{v}$, the message of $\mathsf{SA\mbox{-}ECC}$ can only be one element in the coset corresponding to $\mathsf{v}$. This confirms that the source is a flat distribution on an affine subspace in $\{0,1\}^n$, hence an affine source. The entropy of the affine source is then the dimension of the affine space, which is at least $n-\tau N-\beta$. 
Since this is true for any $\mathsf{V}=\mathsf{v}$, one has the following
$$
(\mathsf{V},\mathsf{aExt}(\mathsf{U}_n))\stackrel{\varepsilon_A}{\sim} (\mathsf{V},\mathsf{U}_\ell),
$$ 
where $\varepsilon_A$ is the error (measured in statistical distance) of the extractor $\mathsf{aExt}$.
Finally, the privacy and leakage-resiliency error is the statistical distance between two views $\mathsf{V}_0$ and $\mathsf{V}_1$ that are corresponding to a pair of secrets $\mathsf{s}_0$ and $\mathsf{s}_1$, respectively. 
One can use the above bound for uniform secret to obtain the following bound for any secret $\mathsf{s}$ (using Lemma~\ref{lem: conditioning} for example).
\begin{equation}\label{eq: exponential error}
(\mathsf{V}|\mathsf{aExt}(\mathsf{U}_n)=\mathsf{s})\stackrel{2^\ell\cdot\varepsilon_A}{\sim} \mathsf{V}.
\end{equation}
Observe that $\mathsf{V}_0=(\mathsf{V}|\mathsf{aExt}(\mathsf{U}_n)=\mathsf{s}_0)$ and $\mathsf{V}_1=(\mathsf{V}|\mathsf{aExt}(\mathsf{U}_n)=\mathsf{s}_1)$. They are both $(2^\ell\cdot\varepsilon_A+\mu)$-close to the distribution of $\mathsf{V}$. 
It then follows that the privacy and leakage-resiliency error is $\varepsilon=2^{\ell+1}\cdot\varepsilon_A+2\mu$.

The construction in Theorem \ref{th: non-adaptive} uses a linear seeded extractor instead of an affine extractor to avoid the exponential grow (from $\varepsilon_A$ to $2^\ell\cdot\varepsilon_A$) of errors (see the proof in \cite{preprint}). But also because of the use of the seeded extractor, which only provides security when the seed is independent of the source, one can only prove privacy and leakage-resiliency for non-adaptive adversary.

\smallskip

We now analyze the information ratio of the non-adaptive affine leakage-resilient secret sharing (we let $t=r-1$ to have a threshold secret sharing) constructed in Theorem \ref{th: non-adaptive} when instantiated with the $\mathsf{SA\mbox{-}ECC}$ from Lemma \ref{th: SA-ECC} and the $\mathsf{Ext}$ from Lemma \ref{lem: Ext}. The secret length is $\ell=n-\tau N-\beta-O(d)$, where the seed length is $d=O(\log^3(2n/\varepsilon))$. The $\mathsf{SA\mbox{-}ECC}$ encodes $d+n$ bits to $N$ bits and with coding rate $R_{ECC}=\rho-\xi$ for a small $\xi$ determined by $\delta$ (satisfying the relation $\delta=\exp(-\Omega(\xi^2N/\log^2 N))$ according to Lemma \ref{th: SA-ECC}). We then have $n=N(\rho-\xi)-d$, resulting in the information ratio   
\begin{align*}
\frac{N/P}{\ell}=\frac{N/P}{n-\tau N-\beta-O(d)}=\frac{N/P}{N(\rho-\xi)-\tau N-\beta-O(d)}=\frac{N/P}{N(\rho-\tau)-\beta-(\xi N+O(d))},
\end{align*}
where by letting $t=r-1$, we have $\rho-\tau=(r-t)/P=1/P$ and hence the information ratio is $\frac{\ell+\beta+o(\ell)}{\ell}$.

\begin{corollary}\label{cor: non-adaptive} There is a non-adaptive $r$-out-of-$P$ statistical $\beta$-bounded affine leakage-resilient secret sharing for any constant $r$ and $P$ with secret length $\ell$ and information ratio $\frac{\ell+\beta+o(\ell)}{\ell}$ and the security parameters $\varepsilon$ and $\delta$ are both negligible in $\ell$. 
\end{corollary}

\subsection{Adaptive Affine Leakage-Resilient Secret Sharing}
We now provide a different way of reducing the explosion of error in (\ref{eq: exponential error}), which does not sacrifice resiliency against an adaptive adversary.

We first recall a classical framework of constructing seedless extractors from seeded extractors. Seeded extractors are known to explicitly extract all the entropy and are not restricted by source structures. Moreover, there are known constructions of {\em linear} seeded extractors perform almost as well as the best seeded extractors. The elegant idea of this framework is to use a seedless extractor to extract a short output from the structured source, which then serves as the seed for a seeded extractor to extract all the entropy from the same source. For this idea to work, the dependence of the extracted seed on the source has to be carefully analyzed (and removed). 

\begin{lemma}[\cite{moremileage}]\label{lem: closedness} Let $\mathcal{C}$ be a class of distributions over $\{0,1\}^n$. Let $\mathsf{E}:\{0,1\}^n\rightarrow\{0,1\}^d$ be a seedless extractor for $\mathcal{C}$ with error $\epsilon$. Let $\mathsf{F}:\{0,1\}^d\times\{0,1\}^n\rightarrow\{0,1\}^m$. Let $\mathsf{X}$ be a distribution in $\mathcal{C}$ and assume that for every $\mathsf{z}\in\{0,1\}^d$ and $\mathsf{y}\in\{0,1\}^m$, the distribution $(\mathsf{X}|\mathsf{F}(\mathsf{z},\mathsf{X})=\mathsf{y})$ belongs to $\mathcal{C}$. Then 
$$
\mathsf{SD}(\mathsf{E}(\mathsf{X}),\mathsf{F}(\mathsf{E}(\mathsf{X}),\mathsf{X});\mathsf{U}_d,\mathsf{F}(\mathsf{U}_d,\mathsf{X}))\leq 2^{d+3}\epsilon.
$$
\end{lemma}

An example of such a class of distributions is the affine source, in which case we can use an {\em affine} extractor $\mathsf{F}=\mathsf{aExt}$ and a {\em linear} seeded extractor $\mathsf{E}=\mathsf{Ext}$. An affine source $\mathsf{X}$ conditioned on $\mathsf{Ext}(\mathsf{z},\mathsf{X})=\mathsf{y}$, which amounts to a set of linear equations, is still an affine source for $\mathsf{aExt}$. With appropriate choice of parameters, we obtain a better affine extractor $\mathsf{aExt}'(\mathsf{X})\colon=\mathsf{Ext}(\mathsf{aExt}(\mathsf{X}),\mathsf{X})$. 
With an increase of $d$ bits in the input, we have the following invertible affine extractor.
$$\mathsf{aExt}''(\mathsf{Sd}||\mathsf{X})\colon=\mathsf{Ext}(\mathsf{aExt}(\mathsf{X})+\mathsf{Sd},\mathsf{X}),
 $$ 
whose inverter is
$
(\mathsf{aExt}'')^{-1}(\mathsf{s})\colon=\left(\mathsf{aExt}(\mathsf{Ext}^{-1}(\mathsf{Z},\mathsf{s}))+\mathsf{Z}||\mathsf{Ext}^{-1}(\mathsf{Z},\mathsf{s})\right),
$ 
where $\mathsf{Z}\stackrel{\$}{\leftarrow}\{0,1\}^d$. 


\begin{theorem}\label{th: low}
Let $\mathsf{aExt}\colon\{0,1\}^n\rightarrow\{0,1\}^d$ be a $(n-\tau N-\beta-\ell,\varepsilon_A)$-affine extractor. Let $\mathsf{Ext}\colon\{0,1\}^d\times\{0,1\}^n\rightarrow\{0,1\}^\ell$ be a linear $(n-\tau N-\beta-d,\varepsilon_E)$-strong extractor with $\varepsilon_E<\frac{1}{8}$.
Let $\mathsf{SA\mbox{-}ECCenc}\colon\{0,1\}^{d+n}\rightarrow\{0,1\}^N$ be the encoder of a statistical affine erasure correcting code $\mathsf{SA\mbox{-}ECC}$ that corrects $(1-\rho) N$ erasures with error probability $\delta$. Let
$$
\left\{
\begin{array}{ll}
\mathsf{Share}(\mathsf{s})&=\mathsf{SA\mbox{-}ECCenc}(\mathsf{Sd}||\mathsf{X}),\mbox{ where }\mathsf{X}\stackrel{\$}{\leftarrow}\mathsf{Ext}^{-1}(\mathsf{Z},\mathsf{s}) \mbox{ and}\\
                                                  &\ \ \ \mathsf{Sd}=\mathsf{Z}+\mathsf{aExt}(\mathsf{X}) \mbox{ with } \mathsf{Z}\stackrel{\$}{\leftarrow}\{0,1\}^d\\
\mathsf{Recst}(\tilde{\mathsf{y}})&=\mathsf{Ext}(\mathsf{aExt}(\tilde{\mathsf{x}})+\tilde{\mathsf{sd}},\tilde{\mathsf{x}}) \mbox{, where }(\tilde{\mathsf{sd}}||\tilde{\mathsf{x}})=\mathsf{SA\mbox{-}ECCdec}(\tilde{\mathsf{y}}),\\
\end{array}
\right.
$$
Here $\tilde{\mathsf{y}}$ denotes an incomplete version of a share vector $\mathsf{y}\in\{0,1\}^N$ with some of its components replaced by erasure symbols.
Let $\varepsilon=2^{(\ell+1)+(d+4)+2}\varepsilon_A+8\varepsilon_E$.
Then the coding scheme $(\mathsf{Share},\mathsf{Recst})$ is an adaptive affine leakage-resilient secret sharing with security parameters $\varepsilon$, $\delta$, leakage bound $\beta$ and ramp parameters $(t,r,P)$ such that $\tau=t/P$, $\rho=r/P$.
\end{theorem}



\begin{proof} Reconstruction from any $r$ shares follows from the functionality of $\mathsf{SA\mbox{-}ECC}$ and the invertibility guarantee of the invertible extractor, which insures that any correctly recovered pre-image is mapped back to the original secret.


We next prove privacy and leakage resiliency. 
Consider a uniform secret $\mathsf{U}_\ell$. 
By the uniformity guarantee of the inverter, we have $\mathsf{Share}(\mathsf{U}_\ell)=\mathsf{SA\mbox{-}ECCenc}(\mathsf{Sd}||\mathsf{U}_n)$. Our analysis is done for any fixed $\mathsf{Sd}=\mathsf{sd}$. This captures a stronger adversary who on top of adaptively reading $t$ shares, also has access to $\mathsf{Sd}$ through an oracle. It is easy to see that the fixing of $\mathsf{Sd}=\mathsf{sd}$ does not alter the distribution of the source $\mathsf{U}_n$, which remains uniform over $\{0,1\}^n$. Let $\mathsf{V}\colon=\mathsf{View}_{\mathcal{A}_{\beta,\mathsf{affine}}}^{\mathcal{O}(\mathsf{SA\mbox{-}ECCenc}(\mathsf{sd}||\mathsf{U}_n))}$ denote the view of the adversary $\mathcal{A}_{\beta,\mathsf{affine}}$ on the encoding of a uniform source for the fixed $\mathsf{Sd}=\mathsf{sd}$. 
Let $\mathsf{Z}\colon=\mathsf{aExt}(\mathsf{U}_n)+\mathsf{sd}$ denote the seed of the strong linear extractor $\mathsf{Ext}$. Finally, let $\mathsf{S}\colon=\mathsf{Ext}(\mathsf{Z},\mathsf{U}_n)$. 
We study the random variable tuple $(\mathsf{V},\mathsf{Z},\mathsf{S})$ to complete the proof. 

The pair $(\mathsf{Z},\mathsf{S})|\mathsf{V}=\mathsf{v}$ for any fixed $\mathsf{V}=\mathsf{v}$ is by definition $(\mathsf{aExt}(\mathsf{U}_n)+\mathsf{sd},\mathsf{Ext}(\mathsf{aExt}(\mathsf{U}_n)+\mathsf{sd},\mathsf{U}_n))|\mathsf{V}=\mathsf{v}$.
Since $(\mathsf{U}_n|\mathsf{V}=\mathsf{v})$ is an affine source with at least $n-\tau N-\beta$ entropy, according to Lemma \ref{lem: closedness}, we have 
$$
(\mathsf{Z},\mathsf{S})|\mathsf{V}=\mathsf{v} \stackrel{2^{d+3}\varepsilon_A}{\sim} (\mathsf{U}_d,\mathsf{Ext}(\mathsf{U}_d,\mathsf{U}_n))|\mathsf{V}=\mathsf{v}.
$$ 
Our concern is the relation between $\mathsf{S}$ and $\mathsf{V}$, and therefore would like to further condition on values of $\mathsf{Z}$. In this step, we crucially use the linearity of $\mathsf{Ext}$ and the underlying linear space structure of the affine source $\mathsf{U}_n|\mathsf{V}=\mathsf{v}$ to claim that there is a subset $\mathcal{G}\subset\{0,1\}^d$ of good seeds such that $\mathsf{Pr}[\mathsf{U}_d\in\mathcal{G}]\geq 1-4\varepsilon_E$ and for any $\mathsf{z}\in\mathcal{G}$, the distribution of $\mathsf{Ext}(\mathsf{z},\mathsf{U}_n)|\mathsf{V}=\mathsf{v}$ is exactly uniform. This is true because $\mathsf{Ext}(\mathsf{z},\mathsf{U}_n)|\mathsf{V}=\mathsf{v}$ is an affine source. If its entropy is $\ell$, then it is exactly uniform. If its entropy is less than $\ell$, its statistical distance $\varepsilon_E^\mathsf{z}$ from uniform is at least $\frac{1}{2}$. Using an averaging argument we have that at least $1-4\varepsilon_E$ fraction of the seeds should satisfy $\varepsilon_E^\mathsf{z}<\frac{1}{4}$, and hence $\varepsilon_E^\mathsf{z}=0$.
We then use Lemma \ref{lem: conditioning}  with respect to the event $\mathsf{Z}\in\mathcal{G}$ to claim that 
$$
(\mathsf{S}|(\mathsf{V}=\mathsf{v},\mathsf{Z}\in\mathcal{G}))\ \stackrel{\frac{2^{d+4}\varepsilon_A}{1-4\varepsilon_E}}{\sim}\  (\mathsf{Ext}(\mathsf{U}_d,\mathsf{X})|(\mathsf{V}=\mathsf{v},\mathsf{U}_d\in\mathcal{G})), 
$$ 
where the right hand side is exactly $\mathsf{U}_\ell$. Note that the subset $\mathcal{G}$ is determined by the indices of the $t$ shares and by the leakage adversary $\mathcal{A}_{\beta,\mathsf{affine}}$, hence remains the same for any value of $\mathsf{V}=\mathsf{v}$. We then have 
$$
((\mathsf{V},\mathsf{S})|\mathsf{Z}\in\mathcal{G})\ \stackrel{\frac{2^{d+4}\varepsilon_A}{1-4\varepsilon_E}}{\sim}\  (\mathsf{V},\mathsf{U}_\ell).
$$ 
Another application of Lemma \ref{lem: conditioning} with respect to the event $\mathsf{S}=\mathsf{s}$ gives
$$
(\mathsf{V}|(\mathsf{Z}\in\mathcal{G},\mathsf{S}=\mathsf{s}))\ \stackrel{\frac{2^{(\ell+1)+(d+4)}\varepsilon_A}{1-4\varepsilon_E}}{\sim}\  \mathsf{V}.
$$ 
We finally bound the privacy and leakage-resiliency error as follows. 
$$
\begin{array}{l}
\mathsf{SD}((\mathsf{V}|\mathsf{S}=\mathsf{s}_0);(\mathsf{V}|\mathsf{S}=\mathsf{s}_1))\\
\leq 2\mathsf{SD}((\mathsf{V}|\mathsf{S}=\mathsf{s});\mathsf{V})\\
=2\mathsf{Pr}[\mathsf{Z}\in\mathcal{G}]\cdot \mathsf{SD}((\mathsf{V}|(\mathsf{Z}\in\mathcal{G},\mathsf{S}=\mathsf{s})); \mathsf{V})+2\mathsf{Pr}[\mathsf{Z}\notin\mathcal{G}]\cdot \mathsf{SD}((\mathsf{V}|(\mathsf{Z}\notin\mathcal{G},\mathsf{S}=\mathsf{s})); \mathsf{V})\\
\leq 2\left(1\cdot\frac{2^{(\ell+1)+(d+4)}\varepsilon_A}{1-4\varepsilon_E}+(4\varepsilon_E+\varepsilon_A)\cdot 1\right)\\
<2^{(\ell+1)+(d+4)+2}\varepsilon_A+8\varepsilon_E.
\end{array}
$$
\end{proof}


\begin{remark}\label{rmk: error bound}
Note that in the error bound $2^{(\ell+1)+(d+4)+2}\varepsilon_A+8\varepsilon_E$ above, the exponential term $2^{(\ell+1)+(d+4)+2}$ only appears as the multiplier of $\varepsilon_A$, the error of $\mathsf{aExt}$. There are known constructions of affine extractor that can extract from any constant fraction of entropy with error exponentially small in the entropy (see Lemma \ref{lem: affine}). Instantiate $\mathsf{aExt}$ with such an affine extractor and $\mathsf{Ext}$ with Trevisan's seeded extractor (see Lemma \ref{lem: Ext}), we have an explicit construction that provide negligible error with seed length $d$ negligible in $\ell$. 
This adaptive affine leakage-resilient secret sharing has better information ratio (both constant) than the one constructed using $\mathsf{aExt}$ alone. 
When used alone, one has to make $\mathsf{aExt}$ invertible using a One-Time-Pad trick (see \cite{preprint}) that costs $\ell$ bits increase in the input. So the information ratio is $\frac{(\ell+n)/R_{ECC}}{P\ell}$, where $R_{ECC}$ is the rate of the erasure correcting code. Recall that making $\mathsf{aExt}'(\cdot)=\mathsf{Ext}(\mathsf{aExt}(\cdot),\cdot)$ invertible only costs $d$ bits, which is negligible in $\ell$ if we use the linear seeded extractor from Lemma \ref{lem: Ext}. We then have information ratio $\frac{(d+n)/R_{ECC}}{P\ell}\approx\frac{n/R_{ECC}}{P\ell}$, for the same level of privacy and reconstruction errors.


\end{remark}





\section{Affine Leakage-Resilient Non-Malleable Secret Sharing}\label{sec: reduction}
We now extend our model of leakage-resilient secret sharing to the paradigm of leakage-resilient non-malleable secret sharing initiated in \cite{KMS18}.
Let $\mathsf{V}\colon=\mathsf{View}_{\mathcal{A}_{\beta,\mathsf{affine}}}^{\mathcal{O}(\mathsf{Share}(\mathsf{s}))}$ be the view of an adaptive $\beta$-bounded affine adversary $\mathcal{A}_{\beta,\mathsf{affine}}$ as defined in Definition \ref{def: lrSSS}.
A  {\em $\mathcal{F}$-tampering strategy associated with $\mathcal{A}_{\beta,\mathsf{affine}}$} is a metafunction 
$$
\sigma\colon\left(\{0,1\}^{N/P}\right)^t\times\{0,1\}^\beta\rightarrow\mathcal{F}
$$ 
that takes as input a view $\mathsf{V}=\mathsf{v}$ and outputs a tampering function $f^\mathsf{v}\in\mathcal{F}$. 

\begin{definition} \label{def: nmSSS} For integers $0\leq t<r\leq P$, an adaptive affine leakage-resilient secret sharing with security parameters $\varepsilon(N)$, $\delta(N)$, leakage bound $\beta$ and ramp parameters $(t,r,P)$ is said to be non-malleable with respect to a family $\mathcal{F}$ of tampering functions from $\{0,1\}^N$ to $\{0,1\}^N$, if the following property is satisfied.
Let the secret sharing scheme $(\mathsf{Share},\mathsf{Recst})$ be as follows.
$$
\mathsf{Share}\colon\{0,1\}^{\ell(N)}\times\mathcal{R}\rightarrow\{0,1\}^N,\ P|N,
$$
where $\mathcal{R}$ denote the randomness set, and for any $R\subset \mathcal{P}$ of size $|R|=r$, there is a
$$
\mathsf{Recst}_R\colon\left(\{0,1\}^{N/P}\right)^r\rightarrow\{0,1\}^{\ell(N)}\cup\{\bot\}.
$$ 
\begin{itemize}
\item Non-malleability: 
For any 
adaptive $\beta$-bounded affine leakage adversary $\mathcal{A}_{\beta,\mathsf{affine}}$, any $\mathcal{F}$-tampering strategy $\sigma$ associate with $\mathcal{A}_{\beta,\mathsf{affine}}$, any $R\subset \mathcal{P}$ of size $|R|=r$ and any secret $\mathsf{s}\in\{0,1\}^{\ell(N)}$, define the tampering-experiment
$$
\mathrm{Tamper}_\mathsf{s}^{\mathcal{A}_{\beta,\mathsf{affine}},\sigma,R}=\left\{
\begin{array}{c}
\mathsf{c}\leftarrow\mathsf{Share}(\mathsf{s})\\
\mathsf{v}=\mathsf{View}_{\mathcal{A}_{\beta,\mathsf{affine}}}^{\mathcal{O}(\mathsf{c})},f^\mathsf{v}=\sigma(\mathsf{v}), \tilde{\mathsf{c}}= f^\mathsf{v}(\mathsf{c})\\
\tilde{\mathsf{s}}=\mathsf{Recst}_R(\tilde{\mathsf{c}}_R)\\
\mathrm{Output}\ \tilde{\mathsf{s}}.\\
\end{array}
\right\},
$$
which is a random variable over the randomness of the share algorithm $\mathsf{Share}$.
We say the scheme is $\varepsilon(N)$-non-malleable if for any $\mathcal{A}_{\beta,\mathsf{affine}}$, $\sigma$, $R$ and $\mathsf{s}$, there exists a distribution $\mathcal{D}_{\mathcal{A}_{\beta,\mathsf{affine}},\sigma,R}$ over the set $\{0,1\}^{\ell(N)}\cup\{\bot\}\cup\{\mathsf{same}^{*}\}$ such that 
\begin{equation}\label{eq: nm}
\mathrm{Tamper}_\mathsf{s}^{\mathcal{A}_{\beta,\mathsf{affine}},\sigma,R} \stackrel{\varepsilon(N)}{\sim} \mathsf{Copy}(\mathcal{D}_{\mathcal{A}_{\beta,\mathsf{affine}},\sigma,R},\mathsf{s}),
\end{equation}
where $\mathsf{Copy}(\cdot,\cdot)$ is as defined in (\ref{eq: nmExt}).

\end{itemize}
When it is clear from the context, instead of $\varepsilon(N), \delta(N), \ell(N)$, we write $\varepsilon,\delta,\ell$. 
\end{definition}



The general approach we take in constructing affine leakage-resilient non-malleable secret sharing in this work is to start with our adaptive affine leakage-resilient secret sharing construction in previous section and consider how to strengthen it for providing non-malleability. 

\remove{use later
We first recall the construction of binary SSS against adaptive adversary in~\cite{preprint}, which uses an invertible affine extractor and a linear (more generally {\em stochastic affine}) erasure correcting code. 
Let $\mathsf{aExt}\colon\{0,1\}^n\rightarrow\{0,1\}^\ell$ be a $\mu$-invertible affine extractor and $\mathsf{aExt}^{-1}\colon\{0,1\}^\ell\times\mathcal{R}\rightarrow\{0,1\}^n$ be its inverter that maps an $\mathsf{s}\in\{0,1\}^\ell$ to one of its pre-images chosen uniformly at random. 
Let $\mathsf{ECC}$  
be a linear erasure correcting code with encoder $\mathsf{ECCenc}\colon\{0,1\}^n\rightarrow\{0,1\}^N$ that tolerates $(1-\rho)N$ erasures. 
Given any $0\leq\tau<\rho\leq1$,
the following pair of sharing and reconstructing algorithms give a $(\tau N,\rho N,N)$-ramp scheme with binary shares, if $\mathsf{aExt}$ can extract from affine $(n,n-\tau N)$-source.

$$
\left\{
\begin{array}{ll}
\mathsf{Share}(\mathsf{s})&=\mathsf{ECCenc}(\mathsf{aExt}^{-1}(\mathsf{s}));\\
\mathsf{Recst}_R(\mathsf{y})&=\mathsf{aExt}(\mathsf{ECCdec}_R(\mathsf{y})),
\end{array}
\right.
$$
where $\mathsf{y}=\mathsf{c}_R$ is the projection of a share vector $\mathsf{c}$ on $R\subset[N]$ with $|R|=\rho N$, and $\mathsf{ECCdec}_R$ is the erasure code decoding algorithm with respect to $R$.

}

Recall that the idea behind the constructions of affine leakage-resilient secret sharing in the previous section can be summarized as identifying an affine source and managing the extractor error (see Section 3.1). 
The analysis is focused on the message of the erasure correcting code, which is at the same time the source of the affine extractor $\mathsf{aExt}$. 
The block-wise projection function and the affine leakage function applied to the share vector induces an affine leakage on the source of $\mathsf{aExt}$. For non-malleability, we similarly consider the {\em tampering} on the source of $\mathsf{aExt}$ induced by the share vector tampering using functions from the family $\mathcal{F}$. There are a few factors we need to take into account while mimicking the analysis for leakage-resilience. Firstly,  leakage-resilience is defined only concerning the encoder (here sharing algorithm) of the coding scheme while {\em tamper resilience} (e.g. non-malleability) involves both the encoder and the decoder. In this case, the induced source tampering should take the decoding process (here reconstruction algorithm) into account. Secondly, the reconstruction algorithm of a secret sharing only takes $r$ shares and hence the induced source tampering depends on which $r$ (tampered) shares take part in the reconstruction. Finally, the share vector tampering in Definition \ref{def: nmSSS} is chosen based on the view of the leakage adversary. We should also take that into account. 
We first formerly define the concept of an {\em induced tampering} for analysing secret sharing that uses an erasure correcting code as a building block.  
\begin{definition}\label{def: induced tampering} Let $\mathsf{ECC}$ be a linear erasure correcting code with an encoder $\mathsf{ECCenc}\colon\{0,1\}^{n}\rightarrow\{0,1\}^N$ and a decoding algorithm $\mathsf{ECCdec}$. 
Let $\sigma$ be an $\mathcal{F}$-tampering strategy associate with $\mathcal{A}_{\beta,\mathsf{affine}}$. Let $R\subset \mathcal{P}$ be of size $|R|=r$ and $\Pi_R$ denotes the block-wise projection function on the block index set $R$. The induced tampering $g_{\sigma,R}^{\mathsf{v}}\colon \{0,1\}^{n}\rightarrow\{0,1\}^{n}$ at a particular view value 
$\mathsf{v}$ 
for given $\mathsf{ECC}$, $\sigma$ and $R$ is defined as follows. 
\begin{equation}\label{eq: induced}
g_{\sigma,R}^{\mathsf{v}}\colon=\mathsf{ECCdec}_R\circ \Pi_R\circ f^\mathsf{v}\circ \mathsf{ECCenc},
\end{equation}
where $\sigma(\mathsf{v})=f^\mathsf{v}\in\mathcal{F}$.
\end{definition}


\subsection{Non-Malleable with respect to Affine Tampering}
We are now in a good position to show a reduction from affine leakage-resilient non-malleable secret sharing with respect to $\mathcal{F}_{\mathsf{affine}}$ to affine non-malleable extractors.

\begin{theorem}\label{th: ANME} 
Let $\mathsf{anmExt}\colon\{0,1\}^n\rightarrow\{0,1\}^\ell$ be a $\mu$-invertible affine non-malleable $(n-tN/P-\beta,\varepsilon_A)$-extractor and $\mathsf{anmExt}^{-1}\colon\{0,1\}^\ell\times\mathcal{R}\rightarrow\{0,1\}^n$ be its inverter that maps an $\mathsf{s}\in\{0,1\}^\ell$ to one of its pre-images chosen uniformly at random. 
Let $\mathsf{ECCenc}\colon\{0,1\}^n\rightarrow\{0,1\}^N$ be the encoder of a linear erasure correcting code  $\mathsf{ECC}$ that tolerates $N-rN/P$ erasures with decoding error $\delta$. Let
$$
\left\{
\begin{array}{ll}
\mathsf{Share}(\mathsf{s})&=\mathsf{ECCenc(anmExt^{-1}}(\mathsf{s}))\\
\mathsf{Recst}_R(\mathsf{c}_R)&=\mathsf{anmExt}(\mathsf{ECCdec}_R(\mathsf{c}_R)),
\end{array}
\right.
$$
where $R\subset\mathcal{P}$ with $|R|=r$.
Then the coding scheme $(\mathsf{Share},\mathsf{Recst})$ is an adaptive affine leakage-resilient non-malleable secret sharing with respect to $\mathcal{F}_{\mathsf{affine}}$ with security parameters $\varepsilon=(2^{\ell+1}\varepsilon_A+\mu$, $\delta$, leakage bound $\beta$ and ramp parameters $(t,r,P)$.

\end{theorem}

\begin{proof} 
Reconstruction from any $r$ shares follows trivially from the functionality of $\mathsf{ECC}$. 
We next show privacy and leakage-resiliency. Our analysis starts with sharing a uniform secret. According to the definition of a $\mu$-invertible extractor, we have
\begin{equation}\label{eq: inverter}
\mathsf{U}_n\stackrel{\mu}{\sim}\mathsf{anmExt}^{-1}(\mathsf{U}_\ell).
\end{equation}
Without loss of generality, we will assume the message of the erasure correcting code  $\mathsf{ECC}$ is $\mathsf{U}_n$ at the cost of an increase of $\mu$ in the final error parameter.
For any adaptive 
$\beta$-bounded affine leakage adversary $\mathcal{A}_{\beta,\mathsf{affine}}$, let $\mathsf{V}\colon=\mathsf{View}_{\mathcal{A}_{\beta,\mathsf{affine}}}^{\mathcal{O}(\mathsf{ECCenc}(\mathsf{U}_n))}$ be the view of the adversary on the encoding of a uniform source. Since $\mathsf{ECCenc}$ is a linear function, $\mathsf{V}$ is the image of an affine function. This shows that $(\mathsf{U}_n|\mathsf{V}=\mathsf{v})$ is an affine source with at least $n-tN/P-\beta$ entropy. The affine non-malleable $(n-tN/P-\beta,\varepsilon_A)$-extractor $\mathsf{amnExt}$ is in particular an affine $(n-tN/P-\beta,\varepsilon_A)$-extractor, which yields
$$
((\mathsf{V},\mathsf{anmExt}(\mathsf{U}_n))|\mathsf{V}=\mathsf{v})\stackrel{\varepsilon_A}{\sim}  ((\mathsf{V},\mathsf{U}_\ell)|\mathsf{V}=\mathsf{v}) \mbox{ or simply } (\mathsf{V},\mathsf{anmExt}(\mathsf{U}_n))\stackrel{\varepsilon_A}{\sim}  (\mathsf{V},\mathsf{U}_\ell).
$$
This together with Lemma \ref{lem: conditioning} with respect to the event $\mathsf{anmExt}(\mathsf{U}_n)=\mathsf{s}$ for any secret $\mathsf{s}$ gives a privacy and leakage-resiliency error of $2^{\ell+1}\varepsilon_A$.

We finally show non-malleability.
For any affine tampering strategy $\sigma$ and $R\subset\mathcal{P}$ with $|R|=r$, 
let $\mathsf{W}\colon=g_{\sigma,R}^{\mathsf{V}}(\mathsf{U}_n)$ denote the tampered source of $\mathsf{anmExt}$. According to Definition~\ref{def: induced tampering}, the induced tampering $g_{\sigma,R}^{\mathsf{v}}$ is an affine function for any $\mathsf{V}=\mathsf{v}$. The functionality of the affine non-malleable $(n-tN/P-\beta,\varepsilon_A)$-extractor asserts that there is a distribution $\mathcal{D}_{g_{\sigma,R}^{\mathsf{v}}}$ such that
$$
((\mathsf{anmExt}(\mathsf{W}),\mathsf{anmExt}(\mathsf{U}_n))|\mathsf{V}=\mathsf{v})\stackrel{\varepsilon_A}{\sim}  (\mathsf{Copy}(\mathcal{D}_{g_{\sigma,R}^{\mathsf{v}}},\mathsf{U}_\ell),\mathsf{U}_\ell),
$$
where the two copies of $\mathsf{U}_\ell$ are the same random variable and are independent of $\mathcal{D}_{g_{\sigma,R}^{\mathsf{v}}}$. 

Let $\mathcal{D}_{\mathcal{A}_{\beta,\mathsf{affine}},\sigma,R}$ be the convex combination of $\{\mathcal{D}_{g_{\sigma,R}^{\mathsf{v}}}|\mathsf{v}\in \mathcal{V}\}$  with coefficients $\{\mathsf{Pr}[\mathsf{V}=\mathsf{v}]|\mathsf{v}\in \mathcal{V}\}$, where $ \mathcal{V}$ is the range of the affine leakage function.
We then have
\begin{equation} \label{eq: goal}
(\mathsf{anmExt}(\mathsf{W}),\mathsf{anmExt}(\mathsf{U}_n))\stackrel{\varepsilon_A}{\sim}  (\mathsf{Copy}(\mathcal{D}_{\mathcal{A}_{\beta,\mathsf{affine}},\sigma,R},\mathsf{U}_\ell),\mathsf{U}_\ell),
\end{equation}
where the two copies of $\mathsf{U}_\ell$ are the same random variable and are independent of  $\mathcal{D}_{\mathcal{A}_{\beta,\mathsf{affine}},\sigma,R}$.

Applying Lemma \ref{lem: conditioning} to (\ref{eq: goal}) with respect to the event $\mathsf{anmExt}(\mathsf{U}_n)=\mathsf{s}$ for any secret $\mathsf{s}$ yields
$$
(\mathsf{anmExt}(\mathsf{W})|\mathsf{anmExt}(\mathsf{U}_n)=\mathsf{s})\stackrel{2^\ell\cdot\varepsilon_A}{\sim}  \mathsf{Copy}(\mathcal{D}_{\mathcal{A}_{\beta,\mathsf{affine}},\sigma,R},\mathsf{s}),
$$
where $\mathcal{D}_{\mathcal{A}_{\beta,\mathsf{affine}},\sigma,R}$ is independent of $\mathsf{s}$.

Since the tampering experiment with respect to the tuple $\mathcal{A}_{\beta,\mathsf{affine}},\sigma,R$ and $\mathsf{s}$ is $\mu$-close to $(\mathsf{anmExt}(\mathsf{W})|\mathsf{anmExt}(\mathsf{U}_n)=\mathsf{s})$ according to (\ref{eq: inverter}), we have 
$$
\mathsf{Tamper}_\mathsf{s}^{\mathcal{A}_{\beta,\mathsf{affine}},\sigma,R}\ \stackrel{\mu+2^\ell\cdot\varepsilon_A}{\sim}\   \mathsf{Copy}(\mathcal{D}_{\mathcal{A}_{\beta,\mathsf{affine}},\sigma,R},\mathsf{s}).
$$
\end{proof}



Theorem \ref{th: ANME} gives a clean reduction from an affine leakage-resilient non-malleable secret sharing to an invertible affine non-malleable extractor and a linear code that correct erasures. Note that we can use any explicit constructions of invertible affine non-malleable extractors and erasure correcting codes. Any improvement in the constructions of the building blocks will lead to affine leakage-resilient non-malleable secret sharing with better parameters.

\begin{remark}\label{rmk: parameters}
The constructions of affine non-malleable extractors (Lemma~\ref{lem: ANME} and Lemma~\ref{lem: convert}) require source entropy $n-n^\xi/2$ and have output length $\ell=n^{\Omega(1)}$ with extractor error $\varepsilon_A=2^{-n^{\Omega(1)}}+n2^{-n^\xi/2}$, for some $0<\xi<1$. According to \cite{CL17}, they can be made invertible with $\mu=\varepsilon_A$. This means that the privacy threshold $t$ must satisfy
$$
n-tN/P\geq n-n^\xi/2 \ \stackrel{\tau =t/P}{\longrightarrow}\  
\frac{\tau N}{n}\leq \frac{n^\xi/2}{n}\ \stackrel{\frac{n}{N}\leq 1}{\longrightarrow}\ \tau\leq \frac{n^\xi}{2n},
$$
and the non-malleability error is $(2^\ell+1)\cdot\varepsilon_A$.
The construction 
in \cite{CL17} crucially relies on high entropy of the source (entropy $n-n^\xi/2$). 
This means that the affine non-malleable extractors in \cite{CL17} requires the $\tau=t/P$ to be small, hence a large $P$ for given $t$.
On the other hand, by replacing the linear erasure correcting code $\mathsf{ECC}$ with a stochastic affine code, we can reconstruct the secret with any $\rho$ fraction of share vector with negligible error probability at rate $R_{ECC}=\frac{n}{N}\approx\frac{r}{P}$. And this replacement does not affect the analysis of non-malleability in Theorem \ref{th: ANME}. In particular, the induced tampering $g_{\sigma,R}^\mathsf{v}$ in (\ref{eq: induced}) becomes 
\begin{equation}\tag{\ref{eq: induced}'}
g_{\sigma,R}^{\mathsf{v}}\colon=\mathsf{ECCdec}^{\tilde{\mathsf{r}}}_R\circ \Pi_R\circ f^\mathsf{v}\circ \mathsf{ECCenc}^\mathsf{r},
\end{equation}
where $\mathsf{r}$ and $\tilde{\mathsf{r}}$ denote the randomness of the stochastic code and its tampered version, respectively. But since the stochastic code is affine, which means for any fixing of its randomness $\mathsf{r}$ both $\mathsf{ECCenc}^\mathsf{r}$ and $\mathsf{ECCdec}^{\tilde{\mathsf{r}}}_R$ are affine functions, the induced tampering $g_{\sigma,R}^\mathsf{v}$ is still an affine function. 
This means that we can obtain a scheme with arbitrary relative reconstruction threshold $\rho>\tau$. 
Finally, 
the output length of the affine non-malleable extractor is $\ell=n^{\Omega(1)}$ and the non-malleability error bound from Theorem \ref{th: ANME} is $(2^\ell+1)\cdot\varepsilon_A$. In this case, we can not use all $\ell$ bits for secrets. A way to control the non-malleability error is to use $\ell-a$ bits for the real secret and append $a$ random bits. 
This, however, reduces the secret length.
\end{remark}

\subsection{Non-Malleable with respect to Bit-wise Independent Tampering} \label{sec: high}
We consider strengthening the construction of affine leakage-resilient secret sharing in Theorem \ref{th: low} to obtain affine leakage-resilient non-malleable secret sharing. 
\remove{don't know why
The candidate is a {\em linear seeded non-malleable extractor} $\mathsf{nmExt}\colon\{0,1\}^d\times\{0,1\}^n\rightarrow\{0,1\}^\ell$, which by definition gives independence guarantee for two applications of the extractor, with respect to two {\em different} seeds, to the {\em same} source. 
}
Intuitively, we want to replace the linear seeded extractor $\mathsf{Ext}$ in Theorem~\ref{th: low} with a linear seeded non-malleable extractor $\mathsf{nmExt}$.
Using a seeded non-malleable extractor in the construction of non-malleable codes 
has many challenges (as far as we known this has not been considered in the literature).
First of all, the tampered source and the original source are not the same.  
We should first reduce the different sources situation to a same source situation in order to be able to use the functionality of $\mathsf{nmExt}$. 
Secondly, seeded non-malleable extractors 
allow the seed to be arbitrarily tampered, but impose a condition that the tampered seed should never be the same as the original seed (the seed tampering function has no fixed point). Lemma \ref{lem: closedness} only shows that the original seed and the tampered seed are both uniform and independent of the original source and tampered source, respectively.  But the two seeds could be related in an arbitrary way, for example, collide with any probability. 
When the tampered seed coincides the original seed, we don't have independence guarantee for the two copies of outputs. In fact, they are related. We then exploit this relation and use an AMD pre-coding of the secret to detect the tampering.
Besides the challenges coming from using a seeded non-malleable extractor, to be able to invoke Lemma \ref{lem: closedness}, the tampered source should have enough entropy. But we know the adversary of non-malleable secret sharing can overwrite almost the full share vector and leave a small amount of entropy in the tampered source. Luckily, in this case, we can simply consider the tampered source as a leakage and make the source itself independent of the secret. To address these challenges in a systematic fashion, we define the entropy of an affine function with respect to an affine source and use it to separate our discussion into two cases.

The entropy of a function is the entropy of its output when the input is uniform. 
Recall that our analysis is focused on induced tampering (see Definition~\ref{def: induced tampering}) that is applied to the source of the invertible affine extractor. 
Since the induced tampering $g_{\sigma,R}^\mathsf{v}$ is applied only under the condition that the view value is $\mathsf{v}$, we then have to consider the entropy of a function when its input is not uniform. We consider an extension of the notion and define the entropy of a function $g$ with respect to a source $\mathsf{X}$.

\begin{definition}\label{def: entropy}
The entropy of a function $g$ with respect to a source $\mathsf{X}$ is the quantity $\mathsf{H}_\infty(g(\mathsf{X}))$. 
\end{definition}

From now on, we consider a linear erasure correcting code $\mathsf{ECC}$ with encoder $\mathsf{ECCenc}\colon\{0,1\}^{d+n}\rightarrow\{0,1\}^N$. 
Let the input to $\mathsf{aExt}''$ be $(\mathsf{Sd}||\mathsf{U}_n)$. We refer to the first $d$ bits as the {\em seed indicator} and only consider $\mathsf{U}_n$ as the source of $\mathsf{aExt}''$. In fact, in the security analysis, we always consider a fixed $\mathsf{Sd}=\mathsf{sd}$. For any adaptive $\beta$-bounded affine leakage adversary $\mathcal{A}_{\beta,\mathsf{affine}}$, let $\mathsf{V}\colon=\mathsf{View}_{\mathcal{A}_{\beta,\mathsf{affine}}}^{\mathcal{O}(\mathsf{ECCenc}(\mathsf{sd}||\mathsf{U}_n))}$ denote the view of the adversary on the encoding of a uniform source. We have that $(\mathsf{U}_n|\mathsf{V}=\mathsf{v})$ is an affine source with at least $n-tN/P-\beta$ entropy. For any tampering strategy $f$ and reconstruction set $R\subset[N]$ with $|R|=r$, let 
$$
(\tilde{\mathsf{sd}}||\mathsf{W})\colon=g_{\sigma,R}^{\mathsf{V}}(\mathsf{sd}||\mathsf{U}_n)
$$ 
denote the tampered source of $\mathsf{aExt}''$. According to Definition~\ref{def: induced tampering}, the induced tampering $g_{\sigma,R}^{\mathsf{v}}$ is an affine function for any $\mathsf{V}=\mathsf{v}$. 
We call the entropy of $g_{\sigma,R}^\mathsf{v}$ with respect to the source $(\mathsf{U_n}|\mathsf{V}=\mathsf{v})$ the {\em entropy of $g_{\sigma,R}^\mathsf{v}$} for short. 
The entropy of an affine function $g$ with respect to an affine source $\mathsf{X}$ is equal to the dimension of the support of the affine source $g(\mathsf{X})$. The entropy of $g_{\sigma,R}^\mathsf{v}$ is then an integer. 
It is easier to consider $g_{\sigma,R}^\mathsf{v}$ as a function defined over the support of the distribution $\mathsf{U}_n|\mathsf{V}=\mathsf{v}$ (instead of $\{0,1\}^n$). Then we have that the entropy of $g_{\sigma,R}^\mathsf{v}$
is $\mathsf{H}_\infty(\mathsf{W}|\mathsf{V}=\mathsf{v})=\dim (\mathsf{Im}(g_{\sigma,R}^\mathsf{v}))$. Now the fundamental theorem of linear algebra yields 
\begin{equation}\label{eq: entropy}
n-\mathsf{H}_\infty(\mathsf{V})=\dim (\mathsf{Ker}(g_{\sigma,R}^\mathsf{v}))+\mathsf{H}_\infty(\mathsf{W}|(\mathsf{V}=\mathsf{v})).
\end{equation}
The quantity $\dim (\mathsf{Ker}(g_{\sigma,R}^\mathsf{v}))$ characterizes the remaining entropy of $(\mathsf{U_n}|\mathsf{V}=\mathsf{v})$ after revealing $\mathsf{W}=\mathsf{w}$ for some particular $\mathsf{w}$.

We are now ready to strengthen the linear seeded extractor $\mathsf{Ext}$ in Theorem \ref{th: low} to a linear non-malleable extractor $\mathsf{nmExt}$ and show that this together with an AMD pre-coding of the secret provides non-malleability. 

\begin{theorem}\label{th: high}
Let $\mathsf{aExt}\colon\{0,1\}^n\rightarrow\{0,1\}^d$ be a $(\frac{n-tN/P-\beta}{2}-\ell,\varepsilon_A)$-affine extractor. Let $\mathsf{nmExt}\colon\{0,1\}^d\times\{0,1\}^n\rightarrow\{0,1\}^\ell$ be a linear $(\frac{n-tN/P-\beta}{2}-d,\varepsilon_E)$-strong extractor with error $\varepsilon_E<2^{-(d+3)}$.
Let $\mathsf{ECCenc}\colon\{0,1\}^{d+n}\rightarrow\{0,1\}^N$ be the encoder of a linear erasure correcting code $\mathsf{ECC}$ that corrects $N-rN/P$ erasures with probability $\delta$. Let $(\mathsf{AMDenc},\mathsf{AMDdec})$ be an AMD code with detection error $\varepsilon_{AMD}$. Let
$$
\left\{
\begin{array}{ll}
\mathsf{Share}(\mathsf{s})&=\mathsf{ECCenc}(\mathsf{Sd}||\mathsf{X}),\mbox{ where }\mathsf{X}\stackrel{\$}{\leftarrow}\mathsf{nmExt}^{-1}(\mathsf{Z},\mathsf{AMDenc}(\mathsf{s})) \mbox{ and}\\
                                                  &\ \ \ \mathsf{Sd}=\mathsf{Z}+\mathsf{aExt}(\mathsf{X}) \mbox{ with } \mathsf{Z}\stackrel{\$}{\leftarrow}\{0,1\}^d\\ 
\mathsf{Recst}_R(\mathsf{c}_R)&=\mathsf{AMDdec}(\mathsf{nmExt}(\mathsf{aExt}(\tilde{\mathsf{x}})+\tilde{\mathsf{sd}},\tilde{\mathsf{x}})) \mbox{, where }(\tilde{\mathsf{sd}}||\tilde{\mathsf{x}})=\mathsf{ECCdec}_R(\mathsf{c}_R),\\
\end{array}
\right.
$$
where $R\subset\mathcal{P}$ with $|R|=r$.
Let $\varepsilon=2^{\ell+d+7}\varepsilon_A+4\varepsilon_E+\varepsilon_{AMD}$. 
Then the coding scheme $(\mathsf{Share},\mathsf{Recst})$ is an adaptive affine leakage-resilient non-malleable secret sharing with respect to $\mathcal{F}_{\mathsf{BIT}}$ with security parameters $\varepsilon$, $\delta$, leakage bound $\beta$ and ramp parameters $(t,r,P)$.
\end{theorem}

The proof of Theorem \ref{th: high} is rather involved and is given in Appendix \ref{apdx: long proof}. We provide outline of the proof here.
Recall that our goal is to replace the $\mathsf{anmExt}$ in (\ref{eq: goal}) with an invertible affine extractor  $\mathsf{aExt}''(\mathsf{sd}||\cdot)\colon=\mathsf{nmExt}(\mathsf{aExt}(\cdot)+\mathsf{sd},\cdot)$ constructed from suitable affine extractor $\mathsf{aExt}$ and {\em seeded non-malleable extractor} $\mathsf{nmExt}$ such that there is a distribution $\mathcal{D}_{\mathcal{A}_{\beta,\mathsf{affine}},\sigma,R}$ satisfying
\begin{equation} \tag{\ref{eq: goal}'}
(\mathsf{aExt}''(\tilde{\mathsf{sd}}||\mathsf{W}),\mathsf{aExt}''(\mathsf{sd}||\mathsf{U}_n))\sim  (\mathsf{Copy}(\mathcal{D}_{\mathcal{A}_{\beta,\mathsf{affine}},\sigma,R},\mathsf{U}_\ell),\mathsf{U}_\ell),
\end{equation}
where $(\tilde{\mathsf{sd}}||\mathsf{W})\colon=g_{\sigma,R}^{\mathsf{V}}(\mathsf{sd}||\mathsf{U}_n)$ denote the tampered source of the affine extractor with $\mathsf{V}\colon=\mathsf{View}_{\mathcal{A}_{\beta,\mathsf{affine}}}^{\mathcal{O}(\mathsf{ECCenc}(\mathsf{sd}||\mathsf{U}_n))}$ denoting the view of the adversary $\mathcal{A}_{\beta,\mathsf{affine}}$ on the encoding of a uniform source. In other words, we want the secret $\mathsf{S}\colon=\mathsf{aExt}''(\mathsf{sd}||\mathsf{U}_n)$ to be independent of the tampered outcome $\mathsf{aExt}''(\tilde{\mathsf{sd}}||\mathsf{W})$. Similar to the proof of Theorem \ref{th: ANME}, we proceed by first conditioned on a particular view $\mathsf{V}=\mathsf{v}$. A slight difference is we now need to
discuss two cases according to the entropy $\mathsf{H}_\infty(\mathsf{W}|\mathsf{V}=\mathsf{v})$. 

\begin{enumerate}
\item If the entropy $\mathsf{H}_\infty(\mathsf{W}|\mathsf{V}=\mathsf{v})$ is less than $\frac{n-tN/P-\beta}{2}$, we can prove (\ref{eq: stronger}). Intuitively, if the induced affine tampering function $g_{\sigma,R}^{\mathsf{v}}(\cdot)$ overwrites many bits and the information contained in $\mathsf{W}$ is small enough that we can consider $\mathsf{W}$ as a virtual leakage (together with the real leakage $\mathsf{V}$) and directly argue independence. More concretely, the affine source $\mathsf{U}_n|(\mathsf{V}=\mathsf{v},\mathsf{W}=\mathsf{w})$ has entropy $n-\mathsf{H}_\infty(\mathsf{V})-\mathsf{H}_\infty(\mathsf{W}|\mathsf{V}=\mathsf{v})$, which is at least $n-tN/P-\frac{n-tN/P-\beta}{2}$, big enough for the affine extractor $\mathsf{aExt}''(\cdot)$. 
We then have $(\mathsf{aExt}''(\mathsf{sd}||\mathsf{U}_n))|(\mathsf{V}=\mathsf{v},\mathsf{W}=\mathsf{w}))\sim \mathsf{U}_\ell$ and hence
\begin{equation}\label{eq: stronger}
((\mathsf{W},\mathsf{aExt}''(\mathsf{sd}||\mathsf{U}_n)))|\mathsf{V}=\mathsf{v})\sim ((\mathsf{W},\mathsf{U}_\ell)|\mathsf{V}=\mathsf{v}).
\end{equation}

\item If the entropy $\mathsf{H}_\infty(\mathsf{W}|\mathsf{V}=\mathsf{v})$ is at least $\frac{n-tN/P-\beta}{2}$, our target is (\ref{eq: goal}') and we have enough entropy for generating an independent uniform seed for $\mathsf{nmExt}$ in the term $\mathsf{aExt}''(\tilde{\mathsf{sd}}||\mathsf{W})$. But two differences between seedless and seeded non-malleable extractors prevent us from obtaining (\ref{eq: goal}'), and have to settle for (\ref{eq: weaker}). 
Roughly speaking, we allow the tampered outcome to be related to the original secret $\mathsf{S}\colon=\mathsf{aExt}''(\mathsf{sd}||\mathsf{U}_n)$ in a simple way (thanks to restriction to bit-wise tampering) in the event $\bar{\mathcal{E}}_{g_{\sigma,R}^\mathsf{v}}$, when the tampered seed is the same as the original seed and the security of a seeded non-malleable extractor is not available. 
More concretely,
\begin{equation}\label{eq: weaker}
\left\{
\begin{array}{ll}
((\mathsf{aExt}''(\tilde{\mathsf{sd}}||\mathsf{W}),\mathsf{S})|(\mathsf{V}=\mathsf{v},\mathcal{E}_{g_{\sigma,R}^\mathsf{v}}))&\sim ((\mathsf{aExt}''(\tilde{\mathsf{sd}}||\mathsf{W}),\mathsf{U}_\ell)|(\mathsf{V}=\mathsf{v},\mathcal{E}_{g_{\sigma,R}^\mathsf{v}}))\\
((\mathsf{aExt}''(\tilde{\mathsf{sd}}||\mathsf{W}),\mathsf{S})|(\mathsf{V}=\mathsf{v},\bar{\mathcal{E}}_{g_{\sigma,R}^\mathsf{v}}))&\sim ((\mathsf{S}+\Delta_{g_{\sigma,R}^\mathsf{v}},\mathsf{S})|(\mathsf{V}=\mathsf{v},\bar{\mathcal{E}}_{g_{\sigma,R}^\mathsf{v}})),\\
\end{array}
\right.
\end{equation}
where $\mathcal{E}_{g_{\sigma,R}^\mathsf{v}}$ denotes the event that the tampered seed is different from the original seed, which is solely determined by $g_{\sigma,R}^\mathsf{v}$,  and $\Delta_{g_{\sigma,R}^\mathsf{v}}$ is a distribution determined by $g_{\sigma,R}^\mathsf{v}$ (hence independent of $\mathsf{S}$). 
In the event $\mathcal{E}_{g_{\sigma,R}^\mathsf{v}}$, the reconstructed secret is independent of the original secret.
In the event $\bar{\mathcal{E}}_{g_{\sigma,R}^\mathsf{v}}$, the AMD decoder outputs $\bot$, by definition. 
\end{enumerate}

\smallskip
\noindent
\begin{remark}[On Explicit Constructions of Linear Non-malleable Extractors] 
The only linear non-malleable extractors we found in the literature is an inner product based construction $\mathsf{IP(X, enc(Z))}$, where $\mathsf{IP}(\cdot,\cdot)$ denotes the inner product of vectors over finite field $\mathbb{F}_q$ and $\mathsf{enc(Z)}$ is a specific encoding of the seed $\mathsf{Z}$ \cite{XinLiIP}. Let $q=2^\ell$. We can have a non-malleable extractor that outputs $\ell$ bits with exponentially small error, if the source $\mathsf{X}\leftarrow\mathbb{F}_q^{\frac{n}{\ell}}$ has more than {\em half} entropy rate. This extractor is $\mathbb{F}_2$-linear because for any seed $\mathsf{Z}=\mathsf{z}$, we have $\mathsf{IP(X+X', enc(z))}=\mathsf{IP(X, enc(z))}+\mathsf{IP(X', enc(z))}$. This linear non-malleable extractor's output is a constant fraction of $n$ and error is exponentially small in $n$. 
This extractor requires a source entropy rate bigger than half, which makes it not applicable in our construction since the entropy requirement of $\mathsf{nmExt}$ is $\frac{n-tN/P-\beta}{2}-d<\frac{n}{2}$.
\end{remark}

This entropy rate around half barrier existed in the literature of (non-linear) non-malleable extractor constructions \cite{DLTWZ14}, but was quickly overcome \cite{overcome}, being only a technical barrier (not inherent). We next show that to output a $\Omega(\log n)$ number of uniform bits with negligible error, at most $\phi n$ bits of entropy suffices, for any constant $\phi>0$. This is shown using a probabilistic argument (see Appendix \ref{apdx: existence proof} for its proof) and we leave the explicit construction as an interesting open problem.

We conclude this section by stating an existence result for the linear seeded non-malleable extractors with our required parameters.   

\begin{theorem}\label{th: existence}
For all integers $n,d,m$ and positive parameters $k,\varepsilon$, there is a linear seeded non-malleable $(k,\varepsilon)$-extractor $\mathsf{E}:\{0,1\}^d\times\{0,1\}^n\rightarrow\{0,1\}^m$ provided that
\begin{equation}\label{eq: existence}
\left\{
\begin{array}{ll}
d&\geq\log(n/\varepsilon^2)+O(1),\\
2m&\leq\log(k+\log\varepsilon)-\log(1/\varepsilon^2)-\log d-O(1).
\end{array}
\right.
\end{equation}
\end{theorem}


\section{Conclusion}\label{sec: conclu}
We studied leakage-resilient secret sharing in the non-compartmentalized models and explicitly constructed them for the class of affine leakage functions. The adversary can apply affine leakage functions to the full share vector to obtain the outputs (subject to only a total length bound) as well as outputting any unauthorized set of shares. We gave constructions for non-adaptive adversary and adaptive adversary, respectively. The construction for non-adaptive adversary is near optimal in the sense that the secret length is almost equal to the share length minus the number of leaked bits. We extended our study to make these affine leakage-resilient secret sharing also non-malleable with respect to a family $\mathcal{F}$ of tampering functions. We gave a construction for the family $\mathcal{F}_{\mathsf{affine}}$ of affine tampering functions for secret sharing with low threshold. For the family $\mathcal{F}_{\mathsf{BIT}}$ of Bit-wise Independent Tampering functions, we gave a construction with all choice of threshold. One interesting open question is whether affine leakage and tampering can be studied for secret sharing with arbitrary monotone access structure. Or on the other hand,  whether other non-compartmentalized models can be studied for secret sharing, even the threshold secret sharing. Our results about leakage-resilient non-malleable secret sharing also motivate open questions concerning explicit constructions of randomness extractors, in particular, affine non-malleable extractors and linear seeded non-malleable extractors.



\bibliographystyle{alpha}
\bibliography{nmSSS.bib}



\appendix
\section*{Appendices}
\addcontentsline{toc}{section}{Appendices}
\renewcommand{\thesubsection}{\Alph{subsection}}


\section{Proof for Lemma \ref{lem: conditioning}} \label{apdx: conditioning}
\begin{proof}
Assume by contradiction that $\mathsf{SD}(\mathsf{V}|\mathsf{W}\in\mathcal{E};\mathsf{V}'|\mathsf{W}'\in\mathcal{E})>\frac{2\varepsilon}{\mathsf{Pr}[\mathsf{W}\in\mathcal{E}]}=\varepsilon_0$. 
W.l.o.g. there is an event $\Omega\subset\mathcal{V}$ (complementing $\Omega$ if necessary) , such that 
$$
\mathsf{Pr}[\mathsf{V}\in\Omega|\mathsf{W}\in\mathcal{E}]-\mathsf{Pr}[\mathsf{V}'\in\Omega|\mathsf{W}'\in\mathcal{E}]>\varepsilon_0.
$$
Now consider the event $\Omega\times\mathcal{E}\subset\mathcal{V}\times\mathcal{W}$. We have
$$
\left \{
\begin{array}{ll}
\mathsf{Pr}[(\mathsf{V},\mathsf{W})\in\Omega\times\mathcal{E}]&=\mathsf{Pr}[\mathsf{V}\in\Omega|\mathsf{W}\in\mathcal{E}]\cdot\mathsf{Pr}[\mathsf{W}\in\mathcal{E}];\\
\mathsf{Pr}[(\mathsf{V}',\mathsf{W}')\in\Omega\times\mathcal{E}]&=\mathsf{Pr}[\mathsf{V}'\in\Omega|\mathsf{W}'\in\mathcal{E}]\cdot\mathsf{Pr}[\mathsf{W}'\in\mathcal{E}].\\
\end{array}
\right.
$$
On the other hand, we have $\mathsf{SD}(\mathsf{W};\mathsf{W}')\leq\mathsf{SD}(\mathsf{V},\mathsf{W};\mathsf{V}',\mathsf{W}')\leq \varepsilon$ and hence
$$
\mathsf{Pr}[\mathsf{W}\in\mathcal{E}]\geq\mathsf{Pr}[\mathsf{W}'\in\mathcal{E}]-\varepsilon.
$$
We then can derive the following contradiction.
$$
\begin{array}{l}
\mathsf{Pr}[(\mathsf{V},\mathsf{W})\in\Omega\times\mathcal{E}]-\mathsf{Pr}[(\mathsf{V}',\mathsf{W}')\in\Omega\times\mathcal{E}]\\
\geq\mathsf{Pr}[\mathsf{W}'\in\mathcal{E}]\cdot(\mathsf{Pr}[\mathsf{V}\in\Omega|\mathsf{W}\in\mathcal{E}]-\mathsf{Pr}[\mathsf{V}'\in\Omega|\mathsf{W}'\in\mathcal{E}])-\varepsilon\\
>\mathsf{Pr}[\mathsf{W}'\in\mathcal{E}]\cdot\varepsilon_0-\varepsilon\\
=\varepsilon.
\end{array}
$$
This concludes the proof.
\end{proof}

\section{Proof for Theorem \ref{th: non-adaptive}}\label{apdx: proof of non-adaptive}
The proof of Theorem \ref{th: non-adaptive} will follow naturally from Lemma \ref{th: extractor property}.  We first recall this general property of a linear strong extractor, which is proved in \cite{preprint}. 

\begin{lemma}[\cite{preprint}]\label{th: extractor property}
Let $\ext\colon \zo^d \times \zo^n \to \zo^m$ be a linear strong $(k,\varepsilon)$-extractor. Let $f_A\colon \zo^{d+n} \to \zo^a$ be any affine function with output length $a\leq n-k$. For any $\mathsf{m},\mathsf{m}'\in\{0,1\}^m$, let $(\mathsf{Z},\mathsf{X})=(\mathsf{U}_d,\mathsf{U}_n)|\left(\mathsf{Ext}(\mathsf{U}_d,\mathsf{U}_n)=\mathsf{m}\right)$ and $(\mathsf{Z}',\mathsf{X}')=(\mathsf{U}_d,\mathsf{U}_n)|\left(\mathsf{Ext}(\mathsf{U}_d,\mathsf{U}_n)=\mathsf{m}'\right)$. We have 
\begin{equation}\label{eq: pairwise}
\mathsf{SD}(f_A(\mathsf{Z},\mathsf{X});f_A(\mathsf{Z}',\mathsf{X}'))\leq 8\varepsilon.
\end{equation}
\end{lemma}

With Lemma \ref{th: extractor property} at hand, we are now in a good position to prove Theorem \ref{th: non-adaptive}. 

\begin{proof}[Proof of Theorem \ref{th: non-adaptive}]
The reconstruction from $r$ shares follows trivially from the definition of stochastic erasure correcting code. We now prove the privacy and leakage resiliency.

The sharing algorithm of the scheme (before applying the stochastic affine code) takes a secret, which is a particular extractor output $\mathsf{s}\in\{0,1\}^\ell$, and uniformly samples a seed $\mathsf{z}\in\{0,1\}^d$ of $\mathsf{Ext}$ before uniformly finds an $\mathsf{x}\in\{0,1\}^n$ such that $\mathsf{Ext}(\mathsf{z},\mathsf{x})=\mathsf{s}$. This process of obtaining $(\mathsf{z},\mathsf{x})$ is the same as sampling uniformly and independently $(\mathsf{U}_d,\mathsf{U}_n)\stackrel{\$}{\leftarrow}\{0,1\}^{d+n}$ and then restricting to $\mathsf{Ext}(\mathsf{U}_d,\mathsf{U}_n)=\mathsf{s}$. We define the random variable pair 
\begin{equation}\label{eq: pre-image}
(\mathsf{Z},\mathsf{X}):=(\mathsf{U}_d,\mathsf{U}_n)|\left(\mathsf{Ext}(\mathsf{U}_d,\mathsf{U}_n)=\mathsf{s}\right)
\end{equation} 
and refer to it as the pre-image of $\mathsf{s}$.

Let $\Pi_A:\left(\{0,1\}^{N/P}\right)^P\rightarrow\left(\{0,1\}^{N/P}\right)^t$ be the projection function that maps a share vector to the $t$ shares with index set $A\subseteq\mathcal{P}$ chosen by the non-adaptive adversary. Observe that the combination $(\Pi_A\circ\mathsf{SA\mbox{-}ECCenc}):\{0,1\}^{d+n}\rightarrow\{0,1\}^t$ (for any fixed randomness $\mathsf{r}$ of $\mathsf{SA\mbox{-}ECCenc}$) is an affine function. 
Moreover, for any affine leakage function $l:\{0,1\}^N\rightarrow\{0,1\}^\beta$, the composition $(l\circ\mathsf{SA\mbox{-}ECCenc}):\{0,1\}^{d+n}\rightarrow\{0,1\}^\beta$ is also an affine function.
So the view of the adversary is simply the output of the affine function $f_A=(\Pi_A\circ\mathsf{SA\mbox{-}ECCenc}||l\circ\mathsf{SA\mbox{-}ECCenc})$, where ``$||$'' denotes concatenation, applied to the random variable tuple $(\mathsf{Z},\mathsf{X})$ defined in (\ref{eq: pre-image}). 


We can now formulate the privacy of the scheme in this context. We want to prove that the statistical distance of the views of the adversary for a pair of secrets $\mathsf{s}$ and $\mathsf{s}'$ can be made arbitrarily small. The views of the adversary are the outputs of the affine function $f_A$ with inputs $(\mathsf{Z},\mathsf{X})$ and $(\mathsf{Z}',\mathsf{X}')$ for the secret $\mathsf{s}$ and $\mathsf{s}'$, respectively. 
According to Lemma \ref{th: extractor property}, we then have that the privacy and leakage-resiliency error is $8\times\frac{\varepsilon}{8}=\varepsilon$.
\end{proof}


\section{Proof for Theorem \ref{th: high}}\label{apdx: long proof}
\begin{proof} 
Reconstruction from any $r$ shares follows from the functionality of $\mathsf{ECC}$ and the invertibility guarantee of the invertible extractor, which insures that any correctly recovered pre-image is mapped back to the original secret. 

We next prove non-malleability.
Consider a uniform secret $\mathsf{U}_\ell$. 
By the uniformity guarantee of the inverter, we have $\mathsf{Share}(\mathsf{U}_\ell)=\mathsf{ECCenc}(\mathsf{Sd}||\mathsf{U}_n)$. Our analysis is done for any fixed $\mathsf{Sd}=\mathsf{sd}$. This captures a stronger adversary who on top of adaptively reading $t$ shares, also has access to $\mathsf{Sd}$ through an oracle. It is easy to see that the fixing of $\mathsf{Sd}=\mathsf{sd}$ does not alter the distribution of the source $\mathsf{U}_n$, which remains uniform over $\{0,1\}^n$. Let $\mathsf{V}\colon=\mathsf{View}_{\mathcal{A}_{\beta,\mathsf{affine}}}^{\mathcal{O}(\mathsf{ECCenc}(\mathsf{sd}||\mathsf{U}_n))}$ denote the view of the adversary $\mathcal{A}$ on the encoding of a uniform source. 
Let $(\tilde{\mathsf{sd}}||\mathsf{W})\colon=g_{\sigma,R}^{\mathsf{V}}(\mathsf{sd}||\mathsf{U}_n)$ denote the tampered source of the affine extractor $\mathsf{aExt}''(\mathsf{sd}||\cdot)\colon=\mathsf{Ext}(\mathsf{aExt}(\cdot)+\mathsf{sd},\cdot)$.
Let $\mathsf{Z}\colon=\mathsf{aExt}(\mathsf{U}_n)+\mathsf{sd}$ denote the original seed of $\mathsf{nmExt}$, which is in particular a strong linear extractor. Let $\mathsf{S}\colon=\mathsf{nmExt}(\mathsf{Z},\mathsf{U}_n)$. 
We study the random variable tuple $(\mathsf{V},\mathsf{W},\mathsf{Z},\mathsf{S})$ to complete the proof. 


\begin{enumerate}
\item {\em Handling the low entropy case.} 
We assume the induced tampering $g_{\sigma,R}^\mathsf{v}$ has entropy at most $\frac{n-tN/P-\beta}{2}$. This means that $(\mathsf{U}_n|(\mathsf{V}=\mathsf{v},\mathsf{W}=\mathsf{w}))$ has entropy at least $\frac{n-tN/P-\beta}{2}$, according to (\ref{eq: entropy}). 

The tuple $(\mathsf{Z},\mathsf{S})|(\mathsf{V}=\mathsf{v},\mathsf{W}=\mathsf{w})$ for any fixed  $\mathsf{V}=\mathsf{v}$ and $\mathsf{W}=\mathsf{w}$ is by definition $(\mathsf{aExt}(\mathsf{U}_n)+\mathsf{sd},\mathsf{Ext}(\mathsf{aExt}(\mathsf{U}_n)+\mathsf{sd},\mathsf{U}_n))|(\mathsf{V}=\mathsf{v},\mathsf{W}=\mathsf{w})$.
Since $(\mathsf{U}_n|(\mathsf{V}=\mathsf{v},\mathsf{W}=\mathsf{w}))$ is an affine source with at least $\frac{n-tN/P-\beta}{2}$ entropy, according to Lemma \ref{lem: closedness}, we have 
$$
(\mathsf{Z},\mathsf{S})|(\mathsf{V}=\mathsf{v},\mathsf{W}=\mathsf{w}) \stackrel{2^{d+3}\varepsilon_A}{\sim} (\mathsf{U}_d,\mathsf{Ext}(\mathsf{U}_d,\mathsf{U}_n))|(\mathsf{V}=\mathsf{v},\mathsf{W}=\mathsf{w}).
$$ 
Our concern is the relation between $\mathsf{S}$ and $\mathsf{W}$, and therefore would like to further condition on values of $\mathsf{Z}$. In this step, we crucially use the linearity of $\mathsf{nmExt}$ and the underlying linear space structure of the affine source $(\mathsf{U}_n|(\mathsf{V}=\mathsf{v},\mathsf{W}=\mathsf{w}))$ to claim that there is a subset $\mathcal{G}\subset\{0,1\}^d$ of good seeds such that $\mathsf{Pr}[\mathsf{U}_d\in\mathcal{G}]\geq 1-4\varepsilon_E$ and for any $\mathsf{z}\in\mathcal{G}$, the distribution of $\mathsf{nmExt}(\mathsf{z},\mathsf{U}_n)|(\mathsf{V}=\mathsf{v},\mathsf{W}=\mathsf{w})$ is exactly uniform. This is true because $\mathsf{nmExt}(\mathsf{z},\mathsf{U}_n)|(\mathsf{V}=\mathsf{v},\mathsf{W}=\mathsf{w})$ is an affine source. If its entropy is $\ell$, then it is exactly uniform. If its entropy is less than $\ell$, its statistical distance $\varepsilon_E^\mathsf{z}$ from uniform is at least $\frac{1}{2}$. Using an averaging argument we have that at least $1-4\varepsilon_E$ fraction of the seeds should satisfy $\varepsilon_E^\mathsf{z}<\frac{1}{4}$, and hence $\varepsilon_E^\mathsf{z}=0$.
We then use Lemma \ref{lem: conditioning}  with respect to the event $\mathsf{Z}\in\mathcal{G}$ to claim that 
$$
\mathsf{S}|(\mathsf{V}=\mathsf{v},\mathsf{W}=\mathsf{w},\mathsf{Z}\in\mathcal{G})\ \stackrel{\frac{2^{d+4}\varepsilon_A}{1-4\varepsilon_E}}{\sim}\  \mathsf{nmExt}(\mathsf{U}_d,\mathsf{X})|(\mathsf{V}=\mathsf{v},\mathsf{W}=\mathsf{w},\mathsf{U}_d\in\mathcal{G}), 
$$ 
where the right hand side is exactly $\mathsf{U}_\ell$. Note that the subset $\mathcal{G}$ is determined by the indices of the $t$ shares chosen by the leakage adversary $\mathcal{A}_{\beta,\mathsf{affine}}$ and the induced tampering function $g_{\sigma,R}^{\mathsf{v}}$, hence remains the same for any value of $\mathsf{W}=\mathsf{w}$. We then have 
$$
((\mathsf{W},\mathsf{S})|(\mathsf{V}=\mathsf{v},\mathsf{Z}\in\mathcal{G}))\ \stackrel{\frac{2^{d+4}\varepsilon_A}{1-4\varepsilon_E}}{\sim}\  ((\mathsf{W},\mathsf{U}_\ell)|\mathsf{V}=\mathsf{v}).
$$ 
Another application of Lemma \ref{lem: conditioning} with respect to the event $\mathsf{S}=\mathsf{s}$ gives
$$
(\mathsf{W}|(\mathsf{V}=\mathsf{v},\mathsf{Z}\in\mathcal{G},\mathsf{S}=\mathsf{s}))\ \stackrel{\frac{2^{(\ell+1)+(d+4)}\varepsilon_A}{1-4\varepsilon_E}}{\sim}\  (\mathsf{W}|\mathsf{V}=\mathsf{v}).
$$ 
We finally bound the non-malleability error as follows. 
$$
\begin{array}{l}
\mathsf{SD}(\mathsf{W}|(\mathsf{V}=\mathsf{v},\mathsf{S}=\mathsf{s});(\mathsf{W}|\mathsf{V}=\mathsf{v}))\\
=\mathsf{Pr}[\mathsf{Z}\in\mathcal{G}]\cdot \mathsf{SD}((\mathsf{W}|(\mathsf{V}=\mathsf{v},\mathsf{S}=\mathsf{s},\mathsf{Z}\in\mathcal{G})); (\mathsf{W}|\mathsf{V}=\mathsf{v}))\\
\ \ \ +\mathsf{Pr}[\mathsf{Z}\notin\mathcal{G}]\cdot \mathsf{SD}((\mathsf{W}|(\mathsf{V}=\mathsf{v},\mathsf{S}=\mathsf{s},\mathsf{Z}\notin\mathcal{G})); (\mathsf{W}|\mathsf{V}=\mathsf{v}))\\
\leq1\cdot\frac{2^{(\ell+1)+(d+4)}\varepsilon_A}{1-4\varepsilon_E}+(4\varepsilon_E+\varepsilon_A)\cdot 1\\
<2^{(\ell+1)+(d+4)+1}\varepsilon_A+4\varepsilon_E.
\end{array}
$$

\item {\em Handling the high entropy case.} 
We assume the induced tampering $g_{\sigma,R}^\mathsf{v}$ has entropy at least $\frac{n-tN/P-\beta}{2}$.
 
Note that for any bit-wise independent function $f^\mathsf{v}$, we can define a difference function $\Delta f^\mathsf{v}$ such that for any $\mathsf{c}\in\{0,1\}^N$,
$$
f^\mathsf{v}(\mathsf{c})=\mathsf{c}+\Delta f^\mathsf{v}(\mathsf{c}).
$$
The difference function $\Delta f^\mathsf{v}$ also induces a source tampering $\Delta g_{\sigma,R}^\mathsf{v}$. Now since the erasure correcting code $\mathsf{ECC}$ is linear, we must have for any $\mathsf{m}\in\{0,1\}^{d+n}$,
$$
g_{\sigma,R}^\mathsf{v}(\mathsf{m})=\mathsf{m}+\Delta g_{\sigma,R}^\mathsf{v}(\mathsf{m}).
$$
Let $\Delta\mathsf{W}\colon=\Delta g_{\sigma,R}^\mathsf{V}(\mathsf{sd}||\mathsf{U_n})$ be the tapered source induced by the difference function $\Delta f^\mathsf{v}$. We immediately have 
\begin{equation}\label{eq: separate}
\mathsf{W}=\mathsf{U_n}+\Delta\mathsf{W}.
\end{equation}
Moreover, since the overwrite bit functions of $f^\mathsf{v}$ become non-overwrite bit functions of $\Delta f^\mathsf{v}$, we then have 
$$
\mathsf{H}_\infty(\Delta\mathsf{W}|\mathsf{V}=\mathsf{v})=n-\mathsf{H}_\infty(\mathsf{V})-\mathsf{H}_\infty(\mathsf{W}|\mathsf{V}=\mathsf{v}).
$$ This means that the dimension of the kernel space of $\Delta f^\mathsf{v}$ restricted to the support of $(\mathsf{U_n}|\mathsf{V}=\mathsf{v})$ satisfies the following.
\begin{equation}\label{eq: difference}
\dim (\mathsf{Ker}(\Delta g_{\sigma,R}^\mathsf{v}))=n-\mathsf{H}_\infty(\mathsf{V})-\mathsf{H}_\infty(\Delta\mathsf{W}|\mathsf{V}=\mathsf{v})=\mathsf{H}_\infty(\mathsf{W}|\mathsf{V}=\mathsf{v})\geq\frac{n-tN/P-\beta}{2}.
\end{equation}
The quantity $\dim (\mathsf{Ker}(\Delta g_{\sigma,R}^\mathsf{v}))$ characterises the remaining entropy in $\mathsf{U}_n$ after conditioning on $\mathsf{V}=\mathsf{v}$ and $\Delta\mathsf{W}=\Delta\mathsf{w}$, for any particular $\Delta\mathsf{w}$.

Now since by assumption $\mathsf{H}_\infty(\mathsf{W}|\mathsf{V}=\mathsf{v})\geq\frac{n-tN/P-\beta}{2}$,  Lemma \ref{lem: closedness} says that 
\begin{equation}\label{eq: seeded}
((\mathsf{aExt}(\mathsf{W})+\tilde{\mathsf{sd}},\mathsf{aExt}''(\tilde{\mathsf{sd}}||\mathsf{W}))|\mathsf{V}=\mathsf{v})\stackrel{2^{d+3}\varepsilon_A}{\sim} ((\mathsf{Z}',\mathsf{nmExt}(\mathsf{Z}',\mathsf{W}))|\mathsf{V}=\mathsf{v}),
\end{equation}
where $\mathsf{Z}'$ is a uniform seed independent of $\mathsf{W}$. We next use (\ref{eq: separate}) and the linearity of $\mathsf{nmExt}$ to claim that 
$$
((\mathsf{Z}',\mathsf{nmExt}(\mathsf{Z}',\mathsf{W}))|\mathsf{V}=\mathsf{v})=((\mathsf{Z}',\mathsf{nmExt}(\mathsf{Z}',\mathsf{U_n})+\mathsf{nmExt}(\mathsf{Z}',\Delta\mathsf{W}))|\mathsf{V}=\mathsf{v}).
$$
We next show that the additive term $\mathsf{nmExt}(\mathsf{Z}',\Delta\mathsf{W})$ can be ignored in the subsequent analysis of comparing $\mathsf{nmExt}(\mathsf{Z}',\mathsf{W})$ against $\mathsf{nmExt}(\mathsf{Z},\mathsf{U}_n)$.
Since the remaining entropy in $\mathsf{U}_n$ after conditioning on $\mathsf{V}=\mathsf{v}$ and $\Delta\mathsf{W}=\Delta\mathsf{w}$ 
is at least $\frac{n-tN/P-\beta}{2}$ (see (\ref{eq: difference})), 
we have according to the functionality of $\mathsf{nmExt}$ that
$$
\begin{array}{r}
((\mathsf{Z},\mathsf{nmExt}(\mathcal{T}(\mathsf{Z}),\mathsf{U_n}),\mathsf{nmExt}(\mathsf{Z},\mathsf{U}_n))|(\mathsf{V}=\mathsf{v},\Delta\mathsf{W}=\Delta\mathsf{w}))\ \ \ \ \ \ \ \ \ \ \ \ \ \ \ \ \ \ \ \ \ \ \ \\\
\stackrel{\varepsilon_E}{\sim} 
((\mathsf{Z},\mathsf{nmExt}(\mathcal{T}(\mathsf{Z}),\mathsf{U_n}),\mathsf{U}_\ell)|(\mathsf{V}=\mathsf{v},\Delta\mathsf{W}=\Delta\mathsf{w})),
\end{array}
$$
where $\mathcal{T}(\cdot)$ is a seed tampering function without fixed point. Let $\mathcal{E}_{g_{\sigma,R}^\mathsf{v}}$  denote the event that $\mathsf{Z}\neq\mathsf{Z}'$ and w.l.o.g. assume $0<\mathsf{Pr}\left[\mathcal{E}_{g_{\sigma,R}^\mathsf{v}}\right]<1$. Applying Lemma \ref{lem: conditioning} with respect to the event  $\mathcal{E}_{g_{\sigma,R}^\mathsf{v}}$ yields
$$
\begin{array}{r}
((\mathsf{Z},\mathsf{nmExt}(\mathsf{Z}',\mathsf{U}_n),\mathsf{nmExt}(\mathsf{Z},\mathsf{U}_n))|(\mathsf{V}=\mathsf{v},\Delta\mathsf{W}=\Delta\mathsf{w},\mathcal{E}_{g_{\sigma,R}^\mathsf{v}}))\ \ \ \ \ \ \ \ \ \ \ \ \ \ \ \ \ \ \ \ \ \ \ \\\
\stackrel{\frac{\varepsilon_E}{\mathsf{Pr}\left[\mathcal{E}_{g_{\sigma,R}^\mathsf{v}}\right]}}{\sim} 
((\mathsf{Z},\mathsf{nmExt}(\mathsf{Z}',\mathsf{U}_n),\mathsf{U}_\ell)|(\mathsf{V}=\mathsf{v},\Delta\mathsf{W}=\Delta\mathsf{w},\mathcal{E}_{g_{\sigma,R}^\mathsf{v}})).
\end{array}
$$
Now for any original seed $\mathsf{z}$ and its tampered version $\mathsf{z}'$, we always have that $(\mathsf{nmExt}(\mathsf{z},\mathsf{U}_n)|(\mathsf{V}=\mathsf{v},\Delta\mathsf{W}=\Delta\mathsf{w},\mathcal{E}_{g_{\sigma,R}^\mathsf{v}},\mathsf{nmExt}(\mathsf{z}',\mathsf{U}_n)=\tilde{\mathsf{s}}))$, for any $\tilde{\mathsf{s}}$, is an affine source. Its statistical distance to uniform is then either $0$ or at least $\frac{1}{2}$. Using an averaging argument, we have for at most $\frac{4\varepsilon_E}{\mathsf{Pr}\left[\mathcal{E}_{g_{\sigma,R}^\mathsf{v}}\right]}$ fraction of such seeds, the above statistical distance exceeds $\frac{1}{4}$. Let $\mathcal{B}$ denote these bad seeds. We then have 
$$
\begin{array}{r}
((\mathsf{nmExt}(\mathsf{Z}',\mathsf{W}),\mathsf{nmExt}(\mathsf{Z},\mathsf{U}_n))|(\mathsf{V}=\mathsf{v},\Delta\mathsf{W}=\Delta\mathsf{w},\mathcal{E}_{g_{\sigma,R}^\mathsf{v}},\mathsf{Z}\notin\mathcal{B}))\ \ \ \ \ \ \ \ \ \ \ \ \ \ \ \\\
=
((\mathsf{nmExt}(\mathsf{Z}',\mathsf{W}),\mathsf{U}_\ell)|(\mathsf{V}=\mathsf{v},\Delta\mathsf{W}=\Delta\mathsf{w},\mathcal{E}_{g_{\sigma,R}^\mathsf{v}},\mathsf{Z}\notin\mathcal{B})).
\end{array}
$$
Taking the error that incurs transforming from seedless extractor to seeded extractor (\ref{eq: seeded}) into account, we have that when the event $\mathcal{E}_{g_{\sigma,R}^\mathsf{v}}$ occurs, the non-malleability error is upper bounded as follows.
$$
\begin{array}{ll}
\varepsilon_{\mathcal{E}_{g_{\sigma,R}^\mathsf{v}}}&\leq 1\cdot\frac{2^{(\ell+1)+(d+4)}\varepsilon_A}{1-\frac{4\varepsilon_E}{\mathsf{Pr}\left[\mathcal{E}_{g_{\sigma,R}^\mathsf{v}}\right]}}+\left(\frac{4\varepsilon_E}{\mathsf{Pr}\left[\mathcal{E}_{g_{\sigma,R}^\mathsf{v}}\right]}+\varepsilon_A\right)\cdot 1\\
                                                                         &\leq \frac{2^{(\ell+1)+(d+4)}\varepsilon_A}{1-2^{d+2}\varepsilon_E}+\left(\frac{4\varepsilon_E}{\mathsf{Pr}\left[\mathcal{E}_{g_{\sigma,R}^\mathsf{v}}\right]}+\varepsilon_A\right)\\
                                                                         &<2^{(\ell+2)+(d+4)}\varepsilon_A+\frac{4\varepsilon_E}{\mathsf{Pr}\left[\mathcal{E}_{g_{\sigma,R}^\mathsf{v}}\right]},
\end{array}
$$
where the second inequality follows from the fact that $\mathsf{Pr}\left[\mathcal{E}_{g_{\sigma,R}^\mathsf{v}}\right]\geq 2^{-d}$ once $\mathsf{Pr}\left[\mathcal{E}_{g_{\sigma,R}^\mathsf{v}}\right]>0$ and the last inequality follows from the assumption that $\varepsilon_E<2^{-(d+3)}$.


On the other hand, if the complimentary event $\bar{\mathcal{E}}_{g_{\sigma,R}^\mathsf{v}}$ occurs, then 
$$
\begin{array}{r}
((\mathsf{Z},\mathsf{nmExt}(\mathsf{Z},\mathsf{W}),\mathsf{S})|(\mathsf{V}=\mathsf{v},\Delta\mathsf{W}=\Delta\mathsf{w}))\ \ \ \ \ \ \ \ \ \ \ \ \ \ \ \ \ \ \ \ \ \ \ \\\
= ((\mathsf{Z},\mathsf{S}+\mathsf{nmExt}(\mathsf{Z},\Delta\mathsf{w}),\mathsf{S})|(\mathsf{V}=\mathsf{v},\Delta\mathsf{W}=\Delta\mathsf{w})).
\end{array}
$$
This means that the tampering results in turning $\mathsf{S}$ into $\mathsf{S}+\mathsf{nmExt}(\mathsf{Z},\Delta\mathsf{w})$, where the offset $\mathsf{nmExt}(\mathsf{Z},\Delta\mathsf{w})$ is independent of $\mathsf{S}$. In this case, let $\mathsf{S}$ be the AMD codeword of the real secret with fresh independent encoding randomness. The decoder of the AMD code outputs $\bot$ with $\varepsilon_{AMD}$.
Taking the error that incurs transforming from seedless extractor to seeded extractor (\ref{eq: seeded}) into account, we have that when the complimentary event $\bar{\mathcal{E}}_{g_{\sigma,R}^\mathsf{v}}$ occurs, the non-malleability error is upper bounded as follows.
$$
\varepsilon_{\bar{\mathcal{E}}_{g_{\sigma,R}^\mathsf{v}}}\leq 1\cdot\frac{2^{(\ell+1)+(d+4)}\varepsilon_A}{1-\mathsf{Pr}\left[\mathcal{E}_{g_{\sigma,R}^\mathsf{v}}\right]}+\varepsilon_{AMD}\cdot 1.
$$

Finally, the total non-malleability error is
$$
\begin{array}{ll}
\varepsilon&\leq \mathsf{Pr}\left[\mathcal{E}_{g_{\sigma,R}^\mathsf{v}}\right]\cdot \varepsilon_{\mathcal{E}_{g_{\sigma,R}^\mathsf{v}}}+\left(1-\mathsf{Pr}\left[\mathcal{E}_{g_{\sigma,R}^\mathsf{v}}\right]\right)\cdot\varepsilon_{\bar{\mathcal{E}}_{g_{\sigma,R}^\mathsf{v}}}\\
                 &<\left(2^{(\ell+2)+(d+4)}\varepsilon_A+4\varepsilon_E\right)+\left(2^{(\ell+1)+(d+4)}\varepsilon_A+\varepsilon_{AMD}\right)\\
                 &<2^{\ell+d+7}\varepsilon_A+4\varepsilon_E+\varepsilon_{AMD}.
\end{array}
$$
\end{enumerate}
\end{proof}


\section{Proof for Theorem \ref{th: existence}} \label{apdx: existence proof}
\begin{proof}
We adapt the proof of \cite{DW07} to show the existence of non-malleable extractors that are linear; i.e., the extractor is a linear function for every fixed seed. This will however result in much weaker parameters than non-linear counterparts.

For a function $\mathsf{E}\colon\{0,1\}^d\times\{0,1\}^n\rightarrow\{0,1\}^m$, distinguisher $\mathcal{D}\colon\{0,1\}^d\times\{0,1\}^m\rightarrow\{0,1\}^m$, seed tampering adversary $\mathcal{A}\colon\{0,1\}^d\rightarrow\{0,1\}^d$, and error parameter $\varepsilon$, call an input $\mathsf{x}\in\{0,1\}^n$ {\em bad} for the tuple $(\mathsf{E},\mathcal{A},\mathcal{D})$ if it violates the following condition for a uniform random seed $\mathsf{S}\stackrel{\$}{\leftarrow}\{0,1\}^d$:
$$
|\mathsf{Pr}[\mathcal{D}(\mathsf{S},\mathsf{E}(\mathcal{A}(\mathsf{S}),\mathsf{x}),\mathsf{E}(\mathsf{S},\mathsf{x}))=1]-\mathsf{Pr}[\mathcal{D}(\mathsf{S},\mathsf{E}(\mathcal{A}(\mathsf{S}),\mathsf{x}),\mathsf{U}_m)=1]|\leq \varepsilon.
$$

Let $\mathcal{B}(\mathsf{E},\mathcal{A},\mathcal{D},\varepsilon)$ denote the set of all bad inputs for $(\mathsf{E},\mathcal{A},\mathcal{D})$ for the parameter $\varepsilon$. We have the following.

\begin{lemma}\label{lem: existence}
Suppose $|\mathcal{B}(\mathsf{E},\mathcal{A},\mathcal{D},\varepsilon)|\leq \varepsilon 2^k$ for all distinguishers $\mathcal{D}$ and adversaries $\mathcal{A}$. Then $\mathsf{E}$ is a non-malleable $(k,2\varepsilon)$-extractor.
\end{lemma}

\begin{proof}
Consider any source $\mathsf{X}$ of min-entropy at least $k$, any distinguisher $\mathcal{D}$ and adversary $\mathcal{A}$. Then,
$$
\mathsf{Pr}[\mathsf{X}\in\mathcal{B}(\mathsf{E},\mathcal{A},\mathcal{D},\varepsilon)]\leq|\mathcal{B}(\mathsf{E},\mathcal{A},\mathcal{D},\varepsilon)|2^{-k}\leq\varepsilon.
$$
Let $\Delta\colon=|\mathsf{Pr}[\mathcal{D}(\mathsf{S},\mathsf{E}(\mathcal{A}(\mathsf{S}),\mathsf{X}),\mathsf{E}(\mathsf{S},\mathsf{X}))=1]-\mathsf{Pr}[\mathcal{D}(\mathsf{S},\mathsf{E}(\mathcal{A}(\mathsf{S}),\mathsf{X}),\mathsf{U}_m)=1]|$. We have 
$$
\Delta\leq \mathsf{Pr}[\mathsf{X}\in\mathcal{B}(\mathsf{E},\mathcal{A},\mathcal{D},\varepsilon)] + \varepsilon\leq 2\varepsilon,
$$
where the first inequality follows from the definition of the bad inputs. The result follows.
\end{proof}

Adapting the notation of \cite{DW07}, the Martingale-based argument of \cite{DW07} proves the following:

\begin{lemma}[\cite{DW07}, Implicit in Theorem 37]\label{lem: Martingale}
Let $\mathsf{x}\in\{0,1\}^n$ be fixed and $\mathbf{E}\colon\{0,1\}^d\times\{0,1\}^n\rightarrow\{0,1\}^m$ be any random function such that $\mathbf{E}(\mathsf{s},\mathsf{x})$ is uniformly random and independent for all choices of $\mathsf{s}\in\{0,1\}^d$. Then, for any distinguisher $\mathcal{D}$, adversary $\mathcal{A}$, and error $\varepsilon>0$,
$$
\mathsf{Pr}[\mathsf{x} \mbox{ is bad for }(\mathbf{E},\mathcal{A},\mathcal{D})]\leq 4\exp(-2^{d-4}\varepsilon^2),
$$
where the probability is over the randomness of $\mathbf{E}$.
\end{lemma}

We now consider a random function $\mathbf{E}\colon\{0,1\}^d\times\{0,1\}^n\rightarrow\{0,1\}^m$. 
This time, however, the random function is linear. That is for every seed $\mathsf{s}$, we independently sample a random $m\times n$ matrix $M_\mathsf{s}$ over $\mathbb{F}_2$ and define $\mathbf{E}(\mathsf{s},\mathsf{x})=M_\mathsf{s}\mathsf{x}$. Consider an adversary that perturbs a seed $\mathsf{s}$ to $\mathcal{A}(\mathsf{s})$, and a distinguisher $\mathcal{D}$.

Let $\mathcal{X}\subset\{0,1\}^n$ be any set of size $\varepsilon 2^k$. Then $\mathcal{X}$ must have a subset $I(\mathcal{X})\subset\mathcal{X}$ of size at least 
$\log|\mathcal{X}|=k+\log\varepsilon$ such that the elements of $I(\mathcal{X})$ are linearly independent. This means that the random variables $\mathbf{E}(\mathsf{s},\mathsf{x})$ for all $\mathsf{x}\in I(\mathcal{X})$ and $\mathsf{s}\in\{0,1\}^d$ are jointly independent. In particular, the events ``$\mathsf{x} \mbox{ is bad for }(\mathbf{E},\mathcal{A},\mathcal{D})$'' are also jointly independent for $\mathsf{x}\in I(\mathcal{X})$. Therefore, using Lemma \ref{lem: Martingale},
$$
\begin{array}{ll}
\mathsf{Pr}[\mbox{all }\mathsf{x}\in\mathcal{X} \mbox{ are bad for }(\mathbf{E},\mathcal{A},\mathcal{D})]&\leq\mathsf{Pr}[\mbox{all }\mathsf{x}\in I(\mathcal{X}) \mbox{ are bad for }(\mathbf{E},\mathcal{A},\mathcal{D})]\\
                     &\leq 4^{|I(\mathcal{X})|}\exp(-2^{d-4}\varepsilon^2|I(\mathcal{X})|)\\
                     &<\exp(2|I(\mathcal{X})|-2^{d-4}\varepsilon^2|I(\mathcal{X})|).
\end{array}
$$
Now, using the above bound 
and the fact that $|I(\mathcal{X})|=k+\log\varepsilon$, we have
$$
\begin{array}{ll}
\mathsf{Pr}[|\mathcal{B}(\mathbf{E},\mathcal{A},\mathcal{D},\varepsilon)|> \varepsilon 2^k]&\leq \mathsf{Pr}[(\exists \mathcal{X}): \mbox{all }\mathsf{x}\in\mathcal{X} \mbox{ are bad for }(\mathbf{E},\mathcal{A},\mathcal{D})]\\
                      &<\exp\left((2-2^{d-4}\varepsilon^2)|I(\mathcal{X})|\right)\cdot {2^n \choose |I(\mathcal{X})|}\\
                      &=\exp\left((2-2^{d-4}\varepsilon^2)(k+\log\varepsilon)\right)\cdot{2^n \choose k+\log\varepsilon},
\end{array}
$$
where in the last inequality we have used a union bound over all possibilities of $I(\mathcal{X})$. Now, by using a union bound over all choices of $\mathcal{D}$ and $\mathcal{A}$ and using Lemma \ref{lem: existence}, we conclude that 
$$
\begin{array}{l}
\mathsf{Pr}[\mathbf{E} \mbox{ is not a non-malleable $(k,\varepsilon)$-extractor}]\\
\leq \exp\left((2-2^{d-4}\varepsilon^2)(k+\log\varepsilon)\right)\cdot2^{n(k+\log\varepsilon)+2^{d+2m}+d2^d}.
\end{array}
$$
The right hand side can be made less than $1$, hence ensuring the existence of a linear non-malleable $(k,\varepsilon)$-extractor provided that (\ref{eq: existence}) holds.
\end{proof}

\end{document}